\documentclass[a4,thm-restate]{article}
\usepackage{fullpage}
\usepackage{amsthm}
\usepackage{amsfonts}

\usepackage{color, tikz, mathtools}
\usetikzlibrary{positioning}
\usetikzlibrary{shapes}
\usetikzlibrary{decorations.pathreplacing}
\usepackage{tcolorbox}
\usepackage{exsheets}
\usepackage{tasks}
\usepackage{multicol}

\usepackage[all,cmtip]{xy}
 
\usepackage{enumerate, enumitem} 
\usepackage{hyperref} 
\hypersetup{colorlinks,linkcolor={blue},citecolor={blue},urlcolor={red}}

\usepackage{graphicx}

\newcommand{\Dom}[1]{#1^{-1}}

\newcommand{\ApIndex}{\mathsf{ApInd}}
\newcommand{\Parity}{\mathsf{\oplus}}
\newcommand{\Tribes}{\mathsf{Tribes}}
\newcommand{\AND}{\mathsf{AND}}
\newcommand{\OR}{\mathsf{OR}}
\newcommand{\pOR}{\mbox{Promise-}\mathsf{OR}}

\newcommand{\GSS}{\mathsf{GSS}}

\newcommand{\GTH}{\mathsf{GTH}}

\newcommand{\C}{\mathsf{C}}
\newcommand{\RC}{\mathsf{RC}}

\newcommand{\FC}{\mathsf{FC}}

\newcommand{\EC}{\mathsf{EC}}

\newcommand{\D}{\mathsf{D}}
\newcommand{\R}{\mathsf{R}}
\newcommand{\Rzero}{\mathsf{R}_0}

\newcommand{\Q}{\mathsf{Q}}
\newcommand{\noisyR}{\mathsf{noisyR}}
\newcommand{\RS}{\mathsf{RS}}

\newcommand{\pAdv}{\mathsf{MM}}

\newcommand{\MM}{\mathsf{MM}}
\newcommand{\CMM}{\mathsf{CMM}}

\DeclareMathOperator{\TV}{TV}

\newcommand{\winning}{\omega}
\newcommand{\ns}{\mathrm{ns}}

\newcommand{\CGpub}{\mathsf{CG}^{\text{\upshape pub}}}
\newcommand{\CG}{\mathsf{CG}^{\text{\upshape priv}}}

\newcommand{\CGstar}{\mathsf{CG}^*}

\newcommand{\CGns}{\mathsf{CG}^\ns}

\newcommand{\s}{\mathsf{s}}
\newcommand{\bs}{\mathsf{bs}}

\newcommand{\ket}[1]{|#1\rangle}
\newcommand{\bra}[1]{\langle#1|}
\newcommand{\norm}[1]{\|#1\|}

\newcommand{\HSph}{\mathcal{S}}
\newcommand{\HBall}{\mathcal{B}}

\newcommand{\SMALL}{\mathbf{Small}}
\newcommand{\BOTH}{\mathbf{Both}}
\newcommand{\brac}[1]{\left\{#1\right\}}

\newcommand{\mathify}[1]{\ifmmode{#1}\else\mbox{$#1$}\fi}

\newcommand{\fbs}{\mathrm{fbs}}

\newcommand{\supp}{\mathrm{supp}}

\newcommand{\zone}{\brac{0, 1}}

\usepackage{tcolorbox}

\newtheorem{theorem}{Theorem}[section]
\newtheorem{corollary}[theorem]{Corollary}

\newtheorem{lemma}[theorem]{Lemma}
\newtheorem{proposition}[theorem]{Proposition}
\newtheorem{claim}[theorem]{Claim}
\newtheorem{definition}[theorem]{Definition}

\usepackage{bbm} 
\usepackage{mathtools}

\newcounter{openctr}

\newcommand{\Open}[1]{{\color{black} {\noindent \bf Open Problem  \stepcounter{openctr} \arabic{openctr} : }#1}}

\title{Certificate Games and Consequences for\\ the Classical Adversary Bound}

\author{Sourav Chakraborty\thanks{Indian Statistical Institute, Kolkata}\and
Anna Gál\thanks{University of Texas at Austin}\and
Mika G\"o\"os\thanks{EPFL} \and
Sophie Laplante\thanks{Université Paris Cité, IRIF}\and
Rajat Mittal\thanks{IIT Kanpur}\and
Anupa Sunny\thanks{Université Paris Cité, IRIF}}

\date{}

\begin{document}
\pagenumbering{roman}
\maketitle

\begin{abstract}
We introduce and study Certificate Game complexity, a measure of complexity based on the probability of winning a game where two players are given inputs with different function values and are asked to output some index $i$ such that $x_i\neq y_i$, in a zero-communication setting.

We study four versions of certificate games, namely private coin, public coin, shared entanglement and non-signaling games. 
The public-coin variant of certificate games gives a new characterization of the classical adversary bound, a lower bound on randomized query complexity which was introduced as a classical version of the quantum (non-negative) quantum adversary bound.

We show that complexity in the public coin model (therefore also the classical adversary) is bounded above  by certificate complexity, as well as by expectational certificate complexity (EC)  and sabotage complexity (RS). On the other hand, it is  bounded below by fractional and randomized certificate complexity.
We provide new exponential separations between classical adversary and randomized query complexity for partial functions.

In contrast, the private coin model is bounded from below by zero-error randomized query complexity and above by EC$^2$. 

The quantum measure reveals an interesting and surprising difference between classical and quantum query models. Whereas the public coin certificate game complexity is bounded from above by randomized query complexity, the quantum certificate game complexity can be quadratically larger than  quantum query complexity. We use non-signaling, a notion from quantum information, to give a lower bound of $n$ on the quantum certificate game complexity of the $\OR$ function, whose quantum query complexity is  $\Theta(\sqrt{n})$, then go on to show that this ``non-signaling bottleneck'' applies to all functions with high sensitivity, block sensitivity, fractional block sensitivity, as well as classical adversary. This implies the collapse of all models of certificate games, except private randomness, to the classical adversary bound.

We consider the single-bit version of certificate games, where the inputs of the two players are restricted to having Hamming distance $1$. We prove that the single-bit version of certificate game complexity with shared randomness is equal to sensitivity up to constant factors, thus giving a new characterization of sensitivity.  On the other hand, the single-bit version of certificate game complexity with private randomness is equal to $\lambda^2$, where $\lambda$ is the spectral sensitivity. 

\end{abstract}

\pagebreak
\tableofcontents

\pagebreak 
\setcounter{page}{1}
\pagenumbering{arabic}
\section{Introduction}


There still remains much to be understood about the complexity of Boolean functions and the many complexity measures that are used
to study various models of computation such as certificate complexity, degree, sensitivity, block sensitivity, their variants, to name a few. Some of the  questions we ask about these measures are: 
What separations can be shown between the  measures? Do they have a natural computational interpretation? What properties do they have, for example, do they behave well under composition? How do they behave for symmetric functions?
Since the sensitivity conjecture was resolved \cite{Huang}, one important new goal  is to determine precisely how the larger measures, such as query complexity and certificate complexity, are bounded above by smaller measures such as sensitivity. The best known upper bound on deterministic query complexity is $\D(f) \leq O(\s(f)^6)$~\cite{NS94, Mid04, Huang},
while the best separation is  cubic~\cite{BHT17}. For certificate complexity we know that $\C(f)\leq O(\s(f)^5)$, whereas the best known separation is cubic~\cite{BBGJK21}.
Many more of these upper bounds and separations are listed in the tables of known results in ~\cite{yu19,ABKRT21}.

With these questions in mind,  we introduce 
new complexity measures  based on the Karchmer-Wigderson relation of a Boolean function. This relation was introduced by Karchmer and Wigderson~\cite{KW90} and it has been extensively studied in communication complexity. 
 Let $f:\{0,1\}^n \rightarrow \{0,1\}$ be a Boolean function.
 The  relation $R_f\subseteq f^{-1}(0)\times f^{-1}(1) \times [n]$ is defined as $R_f = \{ (x,y,i) : x_i \neq y_i\}$. 
 (As a matter of convention, 
 $x$ denotes an input in $f^{-1} (0)$ and $y$ denotes an input in $f^{-1} (1)$ unless otherwise stated.)
 Karchmer and Wigderson~\cite{KW90} showed 
 that the communication complexity of $R_f$ is equal to the circuit depth of $f$.
 We study the following 2-player \emph{certificate game},
 where the goal of the players is to solve the Karchmer-Wigderson relation in a zero-communication setting.
 
 \begin{definition}[Certificate game]
 \label{def:cert-games}
 Let $f:\{0,1\}^n \rightarrow \{0,1\}$ be a (possibly partial) Boolean function.
 One player is given $x\in f^{-1}(0)$ and the other player is given $y\in f^{-1}(1)$. Their goal is to produce a common index $i$ such that $x_i\neq y_i$, without any communication.
 \end{definition}
 We look at how well the players can solve this task in  several  zero-communication settings. 
 We consider four models: when they only have private coins, when they share a public random source,
 and when they share an entangled quantum state (also called quantum model) that does not depend upon their inputs. The fourth model allows any non-signaling strategy which we describe in Section~\ref{sec:CGstar-CGns}. In all these models, we consider the probability of success that they can achieve, for the best strategy and worst case input pair. The multiplicative inverse of the winning probability is called the certificate game complexity of the function ($\CG$ for the private coin model, $\CGpub$ for the public coin model, $\CGstar$ for the shared entanglement model 
 and $\CGns$ for the non-signaling model).

 To illustrate how to achieve such a task  without communication, we consider the following simple strategy.
 Let $f$ be a total Boolean function whose 0-certificate complexity is $c_0$ and whose
 1-certificate complexity is $c_1$. Then on input $x$ such that $f(x)=0$,
 Alice can output a random~$i$ in a minimal 0-certificate for~$x$ (similarly for Bob with a minimal 1-certificate for~$y$). Then since the certificates
 intersect,  the probability that they output the same index is at
 least $\frac{1}{ c_0\cdot c_1}$. This shows that  $\CG(f) \leq \C^0(f) \cdot \C^1(f)$. This simple upper bound is tight for many functions including  $\OR$ and Parity, but there are other examples where $\CG(f)$ can be much smaller, and it is interesting to see what other upper and lower bounds can apply. We will also see that access to shared randomness can significantly reduce the complexity.  
 
We show that the certificate game complexity measures 
in the four different models  hold a pivotal position with respect to other measures,  thus making them good candidates for proving strong lower and upper bounds on various measures. The operational interpretation in terms of winning probability of certificate games makes them convenient for proving upper bounds. Furthermore, the public coin  and non-signaling versions  are  linear programs and therefore their dual formulation is convenient for proving lower bounds.

\subsection{Motivation for certificate games}

The two main ingredients in our certificate games are two-player zero-communication games, and the Karchmer-Wigderson relation.
Two-player zero-communication games have been studied in many different contexts. They are called two-prover games in the context of parallel repetition theorems,
 central to the study of PCPs and the Unique Games Conjecture (we don't consider the case where there could be a quantum verifier, which has been studied in some papers). 
They also appear under the name of zero-communication protocols in the context of communication and information  complexity.
Finally, they are known as local or quantum games in the study of quantum nonlocality, an extensive field 
motivated by the study of quantum entanglement and the relative power of quantum over classical behaviors. 
Quantum behaviors are modeled by two parties making measurements on a shared  bipartite quantum state, and in the classical setup, the two parties can share ``hidden variables'', or shared randomness. 
There has been extensive work, for instance, on simulating quantum behaviors with various resources, such as
communication, post-selection, noise and more. There are also strong connections between finding separations between quantum and classical communication complexity, and between quantum and classical  zero-communication  games.
A survey on quantum non-locality can be found in references~\cite{BCPSW14,PV15}, and on the interactions between communication complexity and nonlocality in reference~\cite{BCMW10}. 

The Karchmer-Wigderson relation $R_f$   appears in many  contexts in the study of complexity measures, including the Adversary bound on quantum query complexity, and its variants~\cite{Amb00, SS06}.  
It is key in understanding
how hard a function is and captures the intuition that if one is to distinguish the 0-instances from the 1-instances of a function, then some~$i$ in the relation has to play a key role in computing the function. 
Another measure where the Karchmer-Wigderson relation appears implicitly is 
Randomized certificate complexity ($\RC$) defined by Aaronson~\cite{Aaronson2008QuantumCC}. It was further shown to be equivalent to fractional block sensitivity and fractional certificate complexity ($\FC$)~\cite{Tal13, GSS16}.  
The non-adaptive version  can be viewed as a one-player 
game where the player is given an input~$x$ and should output an index $i$. The player wins against an input~$y$ (with $f(x)\neq f(y)$) if  $x_i\neq y_i$.

\subsection{Our results}

We show that the certificate game complexity measures of a Boolean function $f$ take pivotal roles in understanding the relationships between various other complexity measures like 
randomised query complexity $\R(f)$, zero-error randomized query complexity $\R_0(f)$, certificate complexity $\C(f)$,
and other related measures. Our results also demonstrate the power of shared randomness over private randomness, even in a zero-communication setting. At the same time, our results also illustrate an interesting, and somewhat counter-intuitive, difference between the quantum world and the classical world. Our main results for {total} functions are compiled in Figure~\ref{fig:measures}. {While most of our results also hold for partial functions, for simplicity we don't indicate that in the Figure. Instead we specify in each theorem whether our result holds for partial functions. { If for a statement or theorem it is not explicitly written that it holds for a possibly partial function then we mean the statement or theorem is only known to hold for total functions.}}

Shared entanglement can simulate shared randomness, and shared randomness gives more power to the players compared to  private randomness so 
\[\CGstar(f) \leq \CGpub(f) \leq \CG(f).\]
A natural question that arises is how separated are these measures. In other words, how much advantage does shared randomness give over private randomness and how much advantage does shared entanglement give over shared randomness?
Because of the operational interpretation of certificate game complexity in terms of the winning probability of certificate games, proving upper bounds on certificate game complexity can be achieved by exhibiting a strategy for the game. 
We provide some other techniques to prove lower bounds.

{The classical adversary bound ($\CMM$, Definition~\ref{def:CMM}) which was defined  in \cite{LM08}, as an analog of the quantum adversary method to study randomized query complexity ($\R$), turns out to play a central role in our work. 
The $\CMM$ measure is well studied with a number of different formulations of it, already known \cite{AKPV}. 
We show  that certificate games (with public randomness), gives another formulation of $\CMM$.  This characterization provides new insight that helps to obtain bounds and separations on $\CMM$ that were not known earlier. }

\paragraph{Lower bounds on certificate games with shared entanglement: }  
One surprising result of our work concerns the shared entanglement model.
In order to prove lower bounds for this model, we introduce 
non-signaling certificate games. Non-signaling is a fundamental concept that comes from quantum non-locality; it states that when making a quantum measurement 
the outcome on one side should not leak any information about the  measurement made on the other side.
This ``non-signaling bottleneck'' is shared by all of our certificate game complexity measures. Identifying it turned out to be the key insight which led to  a very strong lower bound on all these measures, including the quantum model, with a single, simple proof, 
not involving any of the technical overhead  inherent to 
the quantum setting.  The simplicity of the proof comes from the fact that the non-signaling model has several equivalent formulations as linear 
programs, and the strength of the bounds comes from  the fact that it captures precisely a fundamental computational bottleneck.
It also  neatly highlights one of the key differences between quantum and classical query models, since the quantum query model somehow averts this  bottleneck. 

{ Our main lower bound result is a simple and elegant proof  (Theorem~\ref{thm:CGns-cAdv}) that for any, possibly partial, Boolean function $f$,
\[\CGns(f) \geq \CMM(f) \] 
which in turn lower bounds  the other three variants of certificate game complexity. 

The idea is that when a strategy satisfies
the non-signaling condition, the marginal distribution of one of the players' output does not depend on the other player's input. 
Therefore, the marginal distribution of one of the players can be used to give a satisfying assignment for $\CMM$ bound.}

It follows from this lower bound that while the quantum query complexity of the $\OR_n$ function%
    \footnote{%
    $\OR_n$ is the $\OR$ of $n$ variables. 
    From Grover's algorithm  \cite{Grover96,BBHT98} we have $\Q(\OR_n)= \sqrt{n}$. On the other hand $\FC(\OR_n) = \Omega(\s(\OR_n)) = \Omega(n)$.} %
is $\Theta(\sqrt{n})$, its quantum certificate game complexity is $\CGstar(\OR_n) = \Theta(n)$. 

\paragraph{Upper bounds on certificate games with shared randomness: }
The fact that $\CGstar$ is lower bounded by $\CMM$ gives us examples (like the $\OR_n$ function) where the quantum query complexity $\Q$, can be quadratically smaller than $\CGstar$.  In other words, a quantum query algorithm that computes the $\OR_n$ function using $\sqrt{n}$ queries, cannot  reveal to players of a certificate game an index where their inputs differ, with probability better than $1/n$, because of the non-signaling constraint on quantum games. This, somewhat surprisingly, contrasts with the randomized setting where the players can run their randomized query algorithm on their respective inputs using the same random bits
and pick a common random query in order to find an index where the inputs differ, with probability $\smash{\frac{1}{\R(f)}}$, for any~$f$.  Thus, we prove (Theorem~\ref{prop:CGpub-R}) that for any, possibly partial, Boolean function~$f$, \[\CGpub(f) \leq O(\R(f)).\]

{In fact we can prove something much stronger. We prove (Theorem~\ref{thm:cgpub-is-fc}) that for any, possibly partial, Boolean function $f$,
\[\CGpub(f) \leq O(\CMM(f)).\]
Combining this with our lower bound result we have a new characterization of $\CMM$.
\[\CMM(f) \leq \CGns(f) \leq \CGstar(f) \leq \CGpub(f) \leq O(\CMM(f)).\]
This gives us a different way of understanding the classical adversary bound through the lens of two-player games. Slightly weaker results, namely $\CGpub = O(\FC)$ (for total functions) and $\CGns = O(\CMM)$, were proved independently by \cite{swagato} and \cite{Krisjanis}.
}

\paragraph{Bounds on certificate games with private randomness: }
The private randomness model of certificate game complexity, $\CG$, is upper bounded by the product of $0$-certificate complexity, $\C^0$, and $1$-certificate complexity, $\C^1$,  and also by the square of $\EC$ (Theorem~\ref{thm:CG-bounds}). On the other hand $\CG$ is lower bounded by $\R_0$. (This follows from~\cite{JKKLSSV20}.)
Therefore, $\Rzero(f) \leq O(\CG(f)) \leq O(\C^0(f)\C^1(f)).$

In fact, $\CG(f)$ can be larger than the arity of the function. This is because, we show ( Theorem~\ref{thm:CG-bounds}) that $\CG(f)$ is lower bounded by the square of the Minimax formulation of the positive adversary bound, $\pAdv(f)$, which sits between $\Q(f)$  and the spectral sensitivity $\lambda(f)$.

\paragraph{Consequences on expectational certificate complexity 
and the classical adversary bound:}

{
The expectational certificate complexity~\cite{JKKLSSV20} was introduced  as a bound that is quadratically related to zero-error query complexity ($\R_0$), that is, $\EC(f) \leq \R_0(f) \leq O(\EC(f)^2)$ for any total Boolean $f$.
}

We show that $\CGpub$ is  bounded above by $\EC(f)$ up to constant factors (Theorem~\ref{thm:CGpub-EC}), so
$$ \CMM\leq O(\CGpub)\leq O(\EC).$$

{
We also extend our result that $\CGpub$ is upper bounded by $\R(f)$ to prove that $\CGpub$ is also upper bounded by the sabotage complexity $\RS(f)$ (Theorem~\ref{thm:sabotage}). This also proves that for any  Boolean function (including partial functions) $$\CMM(f) = O(\RS(f)).$$
 This gives us upper bounds on $\CMM(f)$ that were not known before,
answering a question asked by Ambainis et al~\cite{AKPV}, where they asked for a general limitation (that includes partial functions) on the power of the classical adversary method as a lower bound on randomized query complexity.
}

{
 For total Boolean functions, $\CMM$ is known to be asymptotically equal to $\FC$. Thus for total functions, the measures $\FC$, $\CMM$, $\CGns$, $\CGstar$ and $\CGpub$ are all asymptotically equal. 
 Our upper bound on $\CGpub$ by $\EC$ implies that $\CGpub$ 
 is also upper bounded by certificate complexity $\C$
 (up to constant factors), since $\EC(f) \leq \C(f)$ for total functions) \cite{JKKLSSV20}. 
 We also give a direct proof that $\CGpub(f) \leq O(\C(f))$ 
 for total functions (Theorem~\ref{thm:CGpub-C}) as a ``warmup" to the stronger upper bound by $\EC$.}

Relating $\EC$ with $\CG$ and $\CGpub$ in turn gives us results about the certificate games themselves. 
To be precise, for total functions, $\EC(f) \leq O(\FC(f){\cdot} \sqrt{\s(f)})$~\cite{JKKLSSV20}. Since $\CG(f)\leq O(\EC(f)^2)$ (Theorem~\ref{thm:CG-bounds}), we have (in Corollary~\ref{cor:CG_CGns}) 
 \[\CG(f) \leq O(\CGpub(f)^3) =O(\CMM(f)^3).\]

\noindent\textbf{Composition:}
The $\Tribes_{\sqrt{n}, \sqrt{n}}$ function is a composition of the $\AND_{\sqrt{n}}$ and $\OR_{\sqrt{n}}$ function. It is easy to show using certificate complexity that $\FC(\Tribes_{\sqrt{n}, \sqrt{n}}) = O(\sqrt{n})$, so $\CMM$, $\CGpub$, $\CGstar$ and $\CGns$ do not compose, that is, 
there are Boolean functions
$f$ and $g$ such that the measures for the function $(f\circ g)$ is not asymptotically the same as the product of the measures for $f$ and for $g$.  The question of whether $\CG$ composes is open.

\paragraph{Certificate game complexity for partial  functions: }

While $\Tribes_{\sqrt{n}, \sqrt{n}}$ demonstrates a quadratic gap between $\R$ and $\CGpub$, 
we know the largest gap between $\R$ and $\CGpub$ for total functions is at most cubic (since $\D \leq (\bs)^3$~\cite{BBC+01,Nisan}). 
But for partial  functions the situation is different.
Ben-David and Blais~\cite{bDB20} demonstrated a function, approximate index $\ApIndex$ (Definition~\ref{def:apIndex}), for which there is exponential separation between $\R$ and $\FC$~%
    \footnote{\cite{bDB20} introduced a measure called $\noisyR$ in an attempt to answer the question of whether $\R$ composes, that is, whether $\R(f\circ g) = \Theta(\R(f)\cdot\R(g))$. They studied $\noisyR$ for the  approximate index function $\ApIndex$ {and showed an exponential separation between  $\noisyR$ and $\R$ for this partial function.}}.
We improve this to show that $\CGpub$ of  $\ApIndex$ is at most $O(\log(\R))$ (Theorem~\ref{thm:CGpub-ApIndex}) and hence demonstrate an exponential separation between $\R$ and $\CGpub$ ($\CMM$) for partial Boolean functions. 

We also give a  partial function $f$
such that $\R(f)= \Omega(n)$ and $\CGpub(f) = O(1)$ (Lemma~\ref{lemma:sep_R_and_CGpub}).

\paragraph{Single-bit versions of certificate games: }

Our final set of results is in the context of single-bit versions of certificate games. Single-bit versions of certain complexity measures
were used in early circuit complexity bounds~\cite{K71,K93}.
More recently Aaronson et al.~\cite{ABKRT21} defined single-bit versions of
several formulations  of the  adversary method, and showed that 
they are all equal to the spectral sensitivity $\lambda$.
Informally, single-bit versions of these measures are obtained
by considering the
requirements only with respect to pairs $x,y$ such that 
$f(x)=0$ and $f(y)=1$
 and $x$ and $y$ differ only in a single bit.

We show that the single-bit version of private coin certificate game complexity is equal to $\lambda^2$ (Theorem~\ref{thm:cg1lambda2}). One of our main results  is that  the single-bit version of public coin certificate game complexity, $\CGpub_{[1]}(f)$ is asymptotically equal to sensitivity $\s(f)$ (Theorem~\ref{thm:cgpub1sens}). This gives a new and very different interpretation of sensitivity, which is one of the central complexity  measures in this area. This interpretation of sensitivity in the context of certificate games may give us a handle on resolving the sensitivity-block sensitivity conjecture (which asks if block sensitivity $\bs(f)$ is $O(\s(f)^2)$, and remains open in this stronger form), by trying to construct a strategy for $\CGpub$ using a strategy for $\CGpub_{[1]}$.

\begin{figure}[htbp]
    \resizebox{.95\linewidth}{!}{

\mbox{\xymatrix{
 & &\C^0\C^1 
        \ar@{->}[dll] \ar@{=>}[dr]^{(1)}& &
    (\EC)^2 
        \ar[d]
        \ar@{=>}[dl]_{(2)}\\
  \D 
        \ar[d] && &
    *+[F]{\hyperref[def:CG]\CG}  
        \ar[dlll]_{(3)} \ar@{=>}[ddr]_{(4)}&
     \FC^2  \ar@{->}[ddll] \ar[d]\\
 \Rzero 
        \ar[d] \ar[drr]
    &&&& \FC \sqrt{\s} \ar@{->}[ddll]\\
 \R 
        \ar[dr]
        \ar@/_2pc/[ddd] 
        \ar[d] 
        & &
    \hyperref[def:C]\C
        \ar@{->}[d]     & 
    &
    (\pAdv)^2 
       \ar@/^1pc/[dddddll]
       \ar[dddddddl]\\ 
  \hyperref[sec:CGpub-ApIndex]\noisyR 
        \ar@/_3pc/[dddddrr]  & 
        \hyperref[def:sabotage]\RS \ar@{=>}[dr]^{(5)}&
        \hyperref[def:EC]{\EC} 
        \ar@{=>}[d]^{(6)  }  \\
 & &*+[F]{\hyperref[def:CGpub]\CGpub} 
        \ar@{=>}[dl] \\
 \Q
        \ar[dddd]  & 
    *+[F]{{\hyperref[sec:CGstar-CGns]\CGstar}} 
         \ar@{=>}[d]  \\
 &*+[F]{\hyperref[def:CGns]\CGns}\ar@{=>}[dr]_{(7)}\\
 && \hyperref[def:CMM]\CMM \ar[d]^{(9)}     
          \ar@{=>}[uuu]_{(8)} 
          \ar[ddll] 
          \\ 
&&\hyperref[def:RC]\FC 
        \ar[d] 
        \\\
 \hyperref[def:pAdv]\pAdv 
        \ar[ddr]& &
    \bs  
        \ar[d] &\lambda^2 \ar[dl]&\\
&& 
\s\ar[dl]\\
 &
 \lambda\\
}

}\quad %
\parbox[t]{10cm}{
\begin{enumerate}
\item  Theorem~\ref{thm:CG-bounds}. Separation: $\GSS_1$ (follows from the fact that $\C^1(\GSS_1) = \Theta(n)$ and $\C^0(\GSS_1) = \Theta(n^2)$). Tightness: $\Parity$.
\item  Theorem~\ref{thm:CG-bounds}, Separation: $\OR$, Tightness: $\Parity$.
\item  Implicit in\cite{JKKLSSV20} (Theorem~\ref{thm:CG-bounds}). Separation: $\Parity$, Tightness: $\OR$.
\item Theorem~\ref{thm:CG-bounds} Separation: Pointer function in \cite{ABBLSS} and the cheat sheet version of the $k-$Forrelation function \cite{BS21, ABK16}. Tightness: $\OR$.
\item  Theorem~\ref{thm:sabotage}. Separation: $\Tribes$ (Theorem~\ref{thm:tribes} and $\RS(\Tribes_{\sqrt{n},\sqrt{n}}) = \Theta(n)$ because $\RS$ composes~\cite{BK18}). Tightness: $\Parity$.
\item Theorem~\ref{thm:CGpub-EC}. Separation: OPEN, Tightness: $\Parity$.
\item Theorem~\ref{thm:CGns-cAdv}. 
\item  Theorem~\ref{thm:cgpub-is-fc}.
\item The reverse direction is known to hold for total functions \cite{AKPV}.
\end{enumerate}

\caption{Some known relations among complexity measures for total functions.
An arrow from $A$ to $B$ indicates that for every total Boolean function $f$, $B(f) = O(A(f))$. Double arrows indicate results in this paper, and boxes indicate new complexity measures. Single arrows indicate known results and references are omitted from the diagram for space considerations. Most references can be found in the tables in~\cite{yu19, ABKRT21} and we cite others in later sections. Known relations about $\EC$ are given in~\cite{JKKLSSV20}, and $\FC = O((\pAdv)^2)$ is proven 
in~\cite{ABK21}. 
Fractional certificate complexity $\FC$ is equal to fractional block sensitivity and to randomized certificate complexity RC (up to multiplicative constants).  $\MM$ is the minimax formulation of the positive adversary method. $\MM = O(\CMM)$ is proved in~\cite{KT16}.
}

\label{fig:measures}

}
}
    
\end{figure}


\section{Certificate game complexity}



In this section, we give the formal definitions of our Certificate Game complexity measures.
 
 A two-player game $G$ is given by a relation $R(x,y,a,b)\subseteq \cal{ X}\times \cal{Y} \times \cal{A}\times \cal{B}$, where $x\in \cal{ X}$ is the first player's input, $y\in \cal{Y}$ is the second player's input. The players output a pair of values, $(a,b) \in \cal{A}\times \cal{B}$, and they win if $R(x,y,a,b)$ holds.
 A \emph{deterministic strategy}
is a pair of functions $A: \cal{X}{\rightarrow} \cal{A}$ and $B: \cal{Y}{\rightarrow} \cal{B}$.
A \emph{randomized strategy with private randomness} is the product of two mixed individual  strategies.
A \emph{randomized strategy with shared randomness} is a mixture of pairs of deterministic strategies.

A \emph{quantum} or \emph{shared entanglement strategy} is given by a shared  bipartite state that does not depend on the input, and a family of projective measurements for Alice, indexed by her input, similarly for Bob.
(More general measurements could be considered, but projective measurements suffice~\cite{CHTW04}.) 

For any strategy, we will write $p(a,b|x,y)$ to mean the probability that the players output $(a,b)$ when their inputs are  $x,y$. The marginal distribution of Alice's output is $p(a|x,y)=\sum_b p(a,b|x,y)$, and similarly, $p(b|x,y)=\sum_a  p(a,b|x,y)$ is Bob's marginal distribution.

\emph{Non-signaling} is a notion that comes from quantum games, which says that if players are spatially separated, then they cannot convey information to each other instantaneously. All the types of strategies described above verify the non-signaling condition.
 \begin{definition}[Non-signaling strategy]
 \label{def:ns}
Let $p(a,b|x,y)$ be the probability that players, on input $x,y$  output $a,b$. Then $p$ is non-signaling if
$p(a|x,y)=p(a|x,y')$ and $p(b|x,y)=p(b|x',y)$ for all inputs $x,x',y,y'$ and all outcomes $a,b$. 
 \end{definition}
Since nonsignaling means that Alice's output does not depend on Bob's input, we can write $p(a|x)$ for Alice's marginal distribution, similarly, we will write $p(b|y)$ for Bob.
 
 Surprisingly, non-signaling strategies are characterized by the \emph{affine combinations} of local deterministic strategies that lie in the positive orthant. This has been known since the 1980s 
 \cite{foulisrandall81, randallfoulis81, klayrandallfoulis87,wilce92}. A more recent proof is given in~\cite{pironio05}.
 
 \begin{proposition}
[Characterization of non-signaling strategies]
 \label{prop:ns-affine}
A strategy $p$ is non-signaling if and only if it is given by a family of coefficients $\lambda = \{\lambda_{AB}\}_{AB}$ (not necessarily nonnegative),  $AB$ ranging over  pairs $(A,B)$ of deterministic strategies,  such that
$p(a,b|x,y)=\sum_{AB : A(x)=a, B(y)=b} \lambda_{AB}$, and 
$\lambda$ verifies  $\sum_{AB}\lambda_{AB}=1$, 
and  $\sum_{AB : A(x)=a,B(y)=b} \lambda_{AB} \geq 0$ for all $a,b,x,y$.
 \end{proposition}
 
 Given a Boolean function $f$ on $n$ variables, define a two-player game such that $\mathcal{X} = f^{-1}(0)$, $\mathcal{Y} = f^{-1}(1)$, $\mathcal{A} =\mathcal{B} = [n]$ and $R(x,y,a,b) = 1$ if and only if $a=b$ and $x_a \neq y_a$. Notice that this setting gives rise to a certificate game according to Definition~\ref{def:cert-games}.

\subsection{Certificate games with private coins}
\label{sec:CGdef}

In case of private coins, a randomized strategy for each player amounts to assigning, for every
 input $x\in\zone^n$, a probability $p_{x,i}$
of producing $i$ as 
its outcome, for each $i\in [n]$.

\begin{definition}[Private coin certificate game complexity]
\label{def:CG}
For a (possibly partial) function $f$,
\[\CG(f) = 
\min_{p} \max_{x,y\in f^{-1}(0)
\times f^{-1}(1)} \frac{1}{\winning(p;x,y)},\]
with $p$  a collection of nonnegative
variables $\{p_{x,i}\}_{x,i}$ satisfying, 
$\sum_{i\in [n]} p_{x,i}=1,$ $\forall x{\in} f^{-1}(0) \cup f^{-1}(1)$,
and 
$\winning(p;x,y)= \sum_{i : x_i\neq y_i} p_{x,i}p_{y,i}$ is the probability that both players output a common index $i$ that satisfies $R_f(x,y,i)$.
\end{definition}

\subsection{Certificate games with public coins}

When the players share randomness, a \emph{public-coin randomized strategy} is 
a distribution over pairs $(A,B)$ of deterministic strategies. 
We assign a nonnegative variable $p_{A,B}$ to each strategy and
require that they sum to 1.
We say that \emph{a pair of strategies $(A,B)$ is correct on
$x,y$} if $A(x)=B(y)=i$ and $x_i\neq y_i$.
\begin{definition}[Public coin certificate game complexity]
\label{def:CGpub}
For a (possibly partial) function $f$,
\[\CGpub(f) = 
  \min_{p} \max_{x,y\in f^{-1}(0)
\times f^{-1}(1)} \frac{1}{\winning^{pub}(p;x,y)},\]
where $p$ is a collection of nonnegative variables $
\{p_{A,B}\}_{A,B}$
satisfying $ \sum_{(A,B)}p_{A,B} = 1$
and 
$\winning^{pub}(p;x,y) = \sum_{(A,B)\text{ correct on } x,y} p_{A,B}
$.
\end{definition}

{\subsection{Certificate games with quantum and non-signaling strategies}
\label{sec:CGstar-CGns}
}

Similar to non-local games (see~\cite{CHTW04}), when the players can share a bipartite quantum state, a general strategy  for a certificate game consists of a shared state $\ket{\Psi_{AB}} \in \mathcal{H}_A \otimes \mathcal{H}_B$ between the two players, and two families of projective measurements $M_A = \{M_A(x)\}_ {x \in f^{-1}(0)}$ and $M_B = \{M_B(x)\}_ {x \in f^{-1}(1)}$ made on their respective part of the shared state. Here $\mathcal{H}_A$ and $\mathcal{H}_B$ are the Hilbert spaces of respective players. For each measurement $M_*(x)$, we denote the family of orthogonal projections as $\{P_{*;x,i}\}_{i\in [n]}$ (see~\cite{NC10} for a definition of projective measurements).

We can now define the shared entanglement certificate game complexity of a Boolean function.
\begin{definition}[Shared entanglement certificate game complexity]
\label{def:CGstar} For a (possibly partial) function $f$,
\[
\CGstar(f) = 
 \min_{\ket{\Psi_{AB}},M_A, M_B} \quad 
\max_{x,y\in f^{-1}(0) \times f^{-1}(1)}
\quad  \frac{1}{\winning^{\star}((\ket{\Psi_{AB}},M_A,M_B);x,y)},
\]
where $\winning^{\star}((\ket{\Psi_{AB}},M_A,M_B);x,y)$ is the winning probability of  strategy $(\ket{\Psi_{AB}},M_A,M_B)$ on 
$x,y$,

$$\winning^{\star}((\ket{\Psi_{AB}},M_A,M_B);x,y) = 
  \sum_{i:x_i \neq y_i} 
  \bra{\Psi_{AB}} P_{A;x,i}\otimes P_{B;y,i} \ket{\Psi_{AB}}.$$
\end{definition}

Non-signaling strategies (Definition~\ref{def:ns}) are a generalization of quantum strategies and are useful to give lower bounds on quantum games. They are particularly well-suited when in a given problem, the bottleneck is that shared entanglement cannot allow players to learn any information about each others' inputs. This is the case for the $\OR$ function (Theorem~\ref{thm:CGns-cAdv}).

\begin{definition}[Non-signaling certificate game complexity]
\label{def:CGns}
For a (possibly partial) function~$f$, 
\[\CGns(f) = 
\min_{p} \quad \max_{x,y\in f^{-1}(0)
\times f^{-1}(1)} \quad  \frac{1}{\winning^{\ns}(p;x,y)}\]
where $p$ ranges over all non-signaling strategies (Def.~\ref{def:ns}) and

 \[\winning^\ns(p;x,y) = \sum_{i : x_i\neq y_i} p(i,i|x,y). \]

\end{definition}
This can be expressed as a linear program by using the affine formulation of non-signaling distributions given in Proposition~\ref{prop:ns-affine}.

Since we have considered progressively stronger models, the following holds trivially.
\begin{proposition}
\label{prop:CGhierarchy}
For any (possibly partial) Boolean function $f$,
\[ \CGns(f) \leq \CGstar(f) \leq \CGpub(f) \leq \CG(f).\] 
\end{proposition}


\section{Overview of our techniques}

The main contribution of this paper is to give lower and upper bounds on certificate game complexity in different models: private coin, public coin and shared entanglement. The bounds on private coin certificate game complexity are obtained by manipulating previously known results and use standard techniques.

The principal contribution, in terms of techniques, is in giving upper and lower bounds on certificate game complexity of public coin and shared entanglement model ($\CGpub$ and $\CGstar$). These techniques can naturally be divided into two parts.

\paragraph{Upper bounds.} 
We use three general techniques for upper bounds on  Certificate games. Given a decision tree, the players can pick queries by agreeing on a node of the decision tree in some way. We use this to show that $\CGpub$ is bounded by sabotage complexity.
For certificate-based measures, we use shared randomness and
hash functions to agree on an common index. We provide some details of this framework in Section~\ref{sec:overview-upper}.

Finally, our strongest general upper bound on public-coin certificate game complexity is $\CGpub(f)\leq O(\CMM(f))$ (Theorem~\ref{thm:cgpub-is-fc}), which implies that all of $\CGpub$, $\CGstar$, $\CGns$ are equal (up to constant factors). The idea behind this upper bound is to apply the \emph{correlated sampling} technique~\cite{Hol09,Bavarian20}. In the correlated sampling game, Alice and Bob receive as input distributions $p$ and~$q$, respectively, and their goal is to output, using shared randomness and no communication, samples~$X\sim p$ and $Y\sim q$ so as to maximize the agreement probability $\Pr[X=Y]$. The basic result about this game is that the players can achieve an agreement probability that depends only on the \emph{total variation distance} between $p$ and $q$. We apply this result in order to convert a non-signaling strategy (which is closely related to $\CMM$)---where $p$ and $q$ roughly correspond to the marginal distributions of the strategy---into a public-coin strategy with only constant-factor loss in the winning probability. The details appear in Section~\ref{sec:cgpub-is-cmm}.

\paragraph{Lower bounds.}
Lower bounds on $\CGpub$ can be obtained by taking the dual of its linear programming formulation. 
     For the shared entanglement model, which is not linear, we  turn to more general non-signaling  games. The resulting non-signaling certificate game complexity, $\CGns$, is a lower bound on $\CGstar$. It can be expressed as a linear program and  lower bounds on $\CGstar$ can be obtained by taking the dual of this linear program and constructing feasible solutions for it.

A more detailed overview of these techniques is given in the following sections.

\subsection{Overview of upper bound techniques for $\CGpub$} \label{sec:overview-upper}

To construct a  strategy for a certificate game, the main challenge is to \emph{match} the index of the other side.
In public coin setting, we can take advantage of having access to shared randomness to achieve this task.

We illustrate this idea by constructing a $\CGpub$ strategy for the 
 $\Tribes$  function.

Even though $\Tribes$ is a starting example for us, it already gives an example that separates $\R$ and $\CGpub$, and also implies that, under function composition, $\CGpub$ value is not the product of the $\CGpub$ value of the individual functions. 
We describe the main idea behind the strategy here. 

For the $\Tribes_{k,k}$ function, we want  a strategy that wins the certificate game with probability $\Omega(1/k)$ (instead of the obvious $\Omega(1/k^2)$). The input of $\Tribes_{k,k}$ consists of $k$ blocks of $k$ bits each. We will reduce the general problem to the case when all blocks of Alice's input have a single $0$, and Bob has exactly one block with all $1$'s and Alice and Bob wins when they both can output the unique index $i$ where Alice's bit is $0$ and Bob's bit is $1$. 

Here we discuss this special case. Let us view Alice's input as an array $A$ of $k$ values, specifying the position of the $0$ in each block (each entry is in $\{1,2, \cdots, k\})$. On the other hand, Bob's input can be thought of as an index, say $j$, between $1$ and $k$, identifying his all-$1$ block. Alice wants to find~$j$ and Bob wants to find $A[j]$, so both can output a position where their inputs differ.

First, take the simple case when each entry of Alice's array is distinct. Bob simply picks a random number $r$ and outputs the $r$-th index of the $j$-th block. Alice can use the same $r$ (due to shared randomness), and find the unique $j$ such that $A[j] = r$. Whenever Bob picks $r$ such that $A[j]=r$, they win the game. Probability that a random $r$ matches $A[j]$ is $1/k$.

For the harder case when some of the entries of $A$ coincide, we use the shared randomness to permute entries of each block. This ensures that, with constant probability, we have a unique $j$ such that $A[j] = r$. This gives the required success probability $\Omega(1/k)$.

\paragraph{A framework for upper bounds based on hashing:}
\label{sec:framework}

Let $f:\{0,1\}^n \rightarrow \{0,1\}$ be a
(possibly partial) Boolean function.
Alice is given $x\in f^{-1}(0)$ and Bob is given $y\in f^{-1}(1)$.
Their goal is to produce a common index $i \in [n]$ such that $x_i\neq y_i$.

Let $T \subseteq [n]$ be a set of potential outputs, known to both players, and let $S$ be a finite
set. $T$ and $S$ are fixed in advance as
part of the specification of the strategy (they do not depend on the input,
only on the function~$f$).
Let $A_x \subseteq T$ denote the set of potential outputs of Alice on $x$
that belong to the set $T$, and
$B_y \subseteq T$ denote the set of potential outputs of Bob  on $y$
that belong to the set $T$.
The players proceed as follows:

\begin{tcolorbox}
\begin{enumerate}
    \item Using shared randomness, they select a random mapping $h: T \rightarrow S$.
    \item Using shared randomness, they select a random element $z \in S$.
\item Alice outputs a (possibly random)  element of $h^{-1}(z) \cap A_x$ 
(if this set is empty, she outputs an arbitrary element).
Similarly, Bob outputs a (possibly random)  element of $h^{-1}(z) \cap B_y$
(if this set is empty, he outputs an arbitrary element).
\end{enumerate}
\end{tcolorbox}

This general strategy will be correct with good enough probability,
if the following two conditions can be ensured:

(i) $h^{-1}(z) \cap W$ is not empty, where
$W \subseteq A_x \cap B_y$ denotes the set of correct outputs from
$A_x \cap B_y$, that is, for any $i \in W$, $x_i \neq y_i$.

(ii) $h^{-1}(z)  \cap A_x$ and $h^{-1}(z) \cap B_y$ are ``small enough''.

Note, Condition (i) implies that both sets,
$h^{-1}(z)  \cap A_x$ and $h^{-1}(z) \cap B_y$, are not empty.

We will apply this general framework in several different ways.
We use it for proving that $\CGpub$ is upper bounded by $\C$ and 
even by $\EC$. We also use it to get a strong upper bound for the  approximate index function $\ApIndex$. Finally, we use the hashing framework to prove that the single-bit version of $\CGpub$
characterizes sensitivity up to constant factors.
While each of these proofs fits into the framework we described above, their analysis is technically quite different.

\paragraph{$\CGpub$ strategy for $\ApIndex$:}
We can use the hashing framework to show an exponential separation between $\R$ and $\CGpub$ for Approximate Index, a partial function. The analysis of the strategy reduces to a very natural question: what is the intersection size of two Hamming balls of radius $\frac{k}{2} - \sqrt{k \log{k}}$ whose centers are at a distance $\frac{k}{\log{k}}$? We are able to show that the intersection is at least an  $\Omega(\frac{1}{\sqrt{log{k}}})$ fraction of the total volume of the Hamming ball. This result and the techniques used could be of independent interest.

To bound the intersection size, we focus on the outermost $\sqrt{k}$ layers of the Hamming ball (since they contain a constant fraction of the total volume), and show that for each such layer the intersection contains an  $\Omega(\frac{1}{\sqrt{log{k}}})$ fraction of the elements in that layer.

For a single layer, the intersection can be expressed as the summation of the latter half of a hypergeometric distribution $P_{k,m,r}$ from $\frac{m}{2}$ to $m$ ($m = \frac{k}{\log{k}}$ is the distance between the Hamming Balls and $r$ is the radius of the layer). By using the ``symmetric'' nature of the hypergeometric distribution around $\frac{m}{2}$ for a sufficient range of values (Lemma~\ref{lem:symmetric}), this reduces to showing a concentration result around the expectation with width $\sqrt{m}$ (as the expectation for our choice of parameters is $ \frac{m}{2} - O(\sqrt{m})$). 

We use the standard concentration bound on hypergeometric distribution with width $\sqrt{r}$ and reduce it to the required width $\sqrt{m}$ by noticing a monotonicity property of the hypergeometric distribution (Lemma~\ref{lem:monotonicity}).

\subsection{Overview of lower bound techniques for $\CGpub$, $\CGstar$ and $\CGns$} \label{sec:overview-lower}

In the public coin setting, maximizing the winning probability in the worst case can be written as a linear program. This allows us to write a dual formulation,
so (since it becomes a minimization problem, and we are considering its multiplicative inverse) this form will be more convenient when proving lower bounds. The dual variables $\mu_{x,y}$ can be thought of as a hard distribution on pairs of inputs, and the objective function is the $\mu$-size of the largest set of input pairs where any deterministic strategy is correct. The next two propositions follow by standard LP duality.

\begin{proposition}[Dual formulation of $\CGpub$] \label{prop:CGpub-dual}

For a two-player certificate game $G_f$ corresponding to a (possibly partial) Boolean function $f$,  $ \CGpub(f) = 1/\winning^{pub}(G_f),$
where the winning probability $\winning^{pub}(G_f)$ is given by the following linear program. 
 \begin{align*}  
 \winning^{pub}(G_f)  = &  \quad \min_{\delta, \mu}
\quad  \delta
  \\
  \text{such that }    &
  \sum_{x,y:~A,B \text{ correct on } x,y} 
  \hspace{-2em}
  \mu_{x,y} \leq \delta \quad \text{ for every deterministic strategy } A,B \\
   & \sum_{x,y} \mu_{xy} = 1, \quad  \mu_{x,y} \geq 0,
 \end{align*}
 where $ \mu =\{\mu_{x,y}\}_{x\in f^{-1}(0), \, y \in f^{-1}(1)}$. $A,B$ correct on $x,y$ implies $A(x)=B(x)=i$ and $x_i\neq y_i$. 
 \end{proposition}

To prove lower bounds on $\CGstar$, we cannot proceed in the same way since the value of $\CGstar$ cannot be written as a linear program. However, a key observation is that in many cases (and in all the cases we have considered in this paper), the fundamental bottleneck for proving lower bounds on quantum strategies is the non-signaling property,  which says that in two-player games with shared entanglement, the outcome of one of the player's measurements cannot reveal the other player's input.
This was the original motivation for defining $\CGns$: 
if we only require the non-signaling property of quantum strategies, it suffices to prove a lower bound on $\CGns$, which is a lower bound on $\CGstar$. Using the characterization of non-signaling strategies in terms of an affine polytope (see Proposition~\ref{prop:ns-affine}), we obtain a convenient linear programming formulation for $\CGns$. 

Definition~\ref{def:CGns} shows that the value of $\winning^\ns(G)$ is a linear optimization problem. We compute its dual, a maximization problem, which allows us to prove lower bounds on $\CGns$ and in turn $\CGstar$.

\begin{proposition}[Dual formulation of $\CGns$]
\label{prop:CGns-dual}

For a  certificate game $G$ corresponding to a (possibly partial) Boolean function $f$, 
$ \CGns(f) = 1/\winning^{ns}(G_f),$  where winning probability $\winning^{ns}(G_f)$ can be written as the following linear program.
 \begin{align*}  
\winning^\ns(G_f) = &  \quad \min_{\mu, \gamma, \delta} \quad \delta
  \\
  \text{such that } & \\
  \sum_{x,y:~A,B \text{ correct on } x,y} & \hspace{-2em} \mu_{x,y}+\sum_{x,y}\gamma_{A(x),B(y),x,y} = \delta \quad \text{ for every deterministic strategy } A,B \\
  \sum_{x,y}  \mu_{xy} = & 1, \quad 
  \mu_{x,y} \geq 0, \quad  \gamma_{a,b,x,y} \geq 0,
 \end{align*}
  where $ \mu =\{\mu_{x,y}\}_{x\in f^{-1}(0), \, y \in f^{-1}(1)}$ and $\gamma =\{\gamma_{i,j,x,y}\}_{i,j\in [n], x\in f^{-1}(0), \, y \in f^{-1}(1)}$ .
 \end{proposition}

As a first step, we illustrate how the dual of the non-signaling variant can be used to prove a lower bound on $\CGstar(\pOR_n)$ (Proposition~\ref{prop:CGnsOR}).
The intuition comes from the fact that  any quantum strategy for the certificate game for $\OR$ has to  be non-signaling. Let one of the player have input $x=0^n$, and  the other player have one of $n$ strings $x^{(i)}$ ($x$ with the $i$-th bit flipped). At the end of the game, they  output $i$ with  probability $p=\frac{1}{\CGstar(\pOR)}$. If this probability were bigger than $\frac{1}{n}$, then the player with input $x$ would learn some information about the other player's input. 

The lower bound on the $\OR$ function generalizes to show that block sensitivity is a lower bound on the non-signaling value of the certificate games. We prove an even stronger result, by going
back to the original definition of $\CGns$ (Definition~\ref{def:CGns}) and giving a very simple proof that $\CGns$ is an upper bound on $\CMM$ (Theorem~\ref{thm:CGns-cAdv}).


\section{Preliminaries}


\label{sec:measures}

We define many known complexity measures in this section.
Almost all definitions are given for arbitrary Boolean functions, including partial functions. A few notable exceptions are certificate complexity, sensitivity and block sensitivity. We include additional details and definitions with respect to partial functions for these measures in Section \ref{sec:partialdefs}.
We use the following notation. A total Boolean function $f$ is $f:\zone^n \rightarrow \zone$. Except when noted otherwise, inputs $x\in \zone^n$ are in $f^{-1}(0)$ and inputs $y\in f^{-1}(1)$, and sums over $x$ range over $x\in f^{-1}(0)$, similarly for $y$. For partial functions we use $\Dom{f}$ for $f^{-1}(0) \cup f^{-1}(1)$. 

Indices $i$ range from 1 to $n$ and $x_i$ denotes the $i$th bit of $x$. We write $x^{(i)}$ to mean the string $x$ with the $i$th bit flipped. When not specified, sums over $i$ range over $i\in [n]$.

 \subsection{Query complexity and adversary bounds}

We recall briefly the standard notations and definitions of query complexity for Boolean functions
$f:\zone^n\rightarrow \zone$.
The deterministic query complexity (or decision tree complexity)  $\D(f)$ is the minimum number of queries to bits of an input $x$ required  to compute $f(x)$, in the worst case.
Randomized query complexity, denoted $\R(f)$, is the number of queries needed to compute $f$, in the worst case, with probability at least $2/3$ for all inputs.
Zero-error randomized query complexity, denoted by $\Rzero(f)$,
is the expected number of queries  needed
to compute $f$ correctly on all inputs. The relation $\R(f)\leq \Rzero(f) \leq \D(f)$ holds for all Boolean functions $f$.
It will be useful to think of a randomized decision tree as a probability distribution over deterministic decision trees. When computing the probability of success, the randomness is over the choice of a deterministic tree.

Quantum query complexity, written $\Q(f)$, is the number of quantum queries needed to compute~$f$ 
correctly on all inputs with probability at least $2/3$.

In this paper we will consider the positive adversary method, a lower bound on quantum query complexity. 
It was shown by Spalek and Szegedy~\cite{SS06} that several formulations  were
equivalent, and we use the MinMax formulation   $\pAdv$ here.

\begin{definition}[Positive adversary method, Minimax formulation]
\label{def:pAdv} For any (possibly partial) Boolean function $f$,
$\pAdv(f) = \min_p \max_{ x\in f^{-1}(0), y\in f^{-1}(1)} \frac{1}{\sum_{i:x_i \neq y_i}\sqrt{p_{x,i}p_{y,i}}}$,
where $p$ is taken over all families of nonnegative $p_{x,i} \in \mathbb{R}$ such that for all $x \in \Dom{f}$ (where $f$ is defined), $\sum_{i\in [n]} p_{x,i} = 1$
\end{definition}

The classical adversary bound was introduced in \cite{Aar06, LM08} as a lower bound for randomized query complexity $\R$. It was shown to be equal to fractional certificate complexity $\FC$ (Definition~\ref{def:RC}) for total functions (but can be larger for partial functions) and has many equivalent formulations, given in~\cite{AKPV}.

\begin{definition}[Classical Adversary Bound]
\label{def:CMM}
For any (possibly partial) Boolean function $f$, the minimax formulation of the Classical Adversary Bound is as:
$\CMM(f) = \min_p \max_{\substack{x, y\in S\\ f(x) = 1- f(y)}} \frac{1}{\sum_{i: x_i \ne y_i} \min\{p_x(i), p_y(i)\}}$,
where $p_x$
is a probability distribution over $[n]$.
\end{definition}

\subsection{Certificate complexity and its variants}

Certificate complexity is 
a lower bound on query complexity~\cite{VW85}, for total Boolean functions.

For a total Boolean function~$f$, a \emph{certificate} is a partial assignment of the bits of an input to~$f$ that forces the value of the function to be constant,
regardless of the value of the other bits. A \emph{certificate for input $x$} is a partial assignment consistent with $x$ that is a certificate for $f$.
\begin{definition}
\label{def:C}
For any total Boolean function $f$ and input $x$, 
$\C(f;x)$ is the size of the smallest  certificate for $x$. The certificate complexity of the function is 
$\C(f) = \max \{\C^0(f), \C^1(f)\}$, where $\C^b(f) = \max_{x\in f^{-1}(b)} \{\C(f;x)\}$.
\end{definition}

Randomized certificate complexity was introduced by Aaronson as a randomized version of certificate complexity~\cite{Aaronson2008QuantumCC}, and subsequently shown to be equivalent (up to constant factors) to fractional block sensitivity and fractional certificate complexity~\cite{Tal13, KT16, GSS16}. 

We use the fractional certificate complexity formulation.

\begin{definition}[Fractional certificate complexity]
\label{def:RC}

For any (possibly partial) Boolean function~$f$ 
$\FC(f) = \max_{z \in \Dom{f}} \, \FC(f,z),$ where
$\FC(f,z) =\min_v \sum_i v_{z,i}\ ,$
subject to $ \sum_{i: z_i \neq z'_i}v_{z,i} \geq 1$ for all $z'\in \Dom{f}$ such that $f(z)=1-f(z')$,
with $v$ a collection of variables $v_{z,i} \geq 0$. 
\end{definition}

Another equivalent formulation is,
$\FC(f) = 
\min_{w}\, \max_{\stackrel{z, z' \in \Dom{f}}{f(z) = 1- f(z')}}
\, \frac{\sum_i w_{z,i}}{\sum_{i:z_i\neq z'_i}w_{z,i}} \ ,
$ 
where~$w$ is a collection of non-negative
variables $w_{z,i}$.

Randomized certificate complexity (in its non-adaptive formulation)
can be viewed as a single player game where a player is given an input~$z$ and should output an index $i$ (say with probability $p_{z,i}= \frac{w_{z,i}}{\sum_j w_{z,j}}$). The player wins against an input~$z'$ (with $f(z) = 1- f(z')$) if  $z_i\neq z'_i$. 

Then, $\FC(f)$, for total functions, is (up to constant factors) the multiplicative inverse of the probability of winning the game in the worst case~\cite{Aaronson2008QuantumCC, Tal13, GSS16}. 

Expectational certificate complexity was introduced as a quadratically tight lower bound on $\Rzero$~\cite{JKKLSSV20}.

\begin{definition}[Expectational certificate complexity~\cite{JKKLSSV20}]
\label{def:EC}
For any (possibly partial) Boolean function $f$, 
$\EC(f) = \min_{w} \max_{z \in \Dom{f}} \sum_{i\in [n]}  w_{z,i}$
with $w$ a collection of variables such that $0\leq w_{z,i}\leq 1$ satisfying $\sum_{i : z_i\neq z'_i} w_{z,i}w_{z',i} \geq 1$ for all $z,z'$ s.t. $f(z) = 1-f(z')$. 
\end{definition}
Since the weights are between 0 and 1, we can associate with each $i$ a Bernoulli variable.  The players can sample from each of these variables independently and output the set of indices where the outcome was 1. The constraint says that the expected number of indices $i$ in both sets that satisfy $z_i\neq z'_i$ should be bounded below by 1. The complexity measure is the expected size of the sets.  For  example, for the OR function, a strategy could be as follows. On input $z$, pick the smallest $i$ for which $z_i=1$, output the set  $\{i\}$. If no such $i$ exists, then output the set $[n]$. The (expected) size of the set is $n$ and the (expected) size of the intersection is $1$.

The following relations are known to hold for any total Boolean function $f$.
\begin{proposition}[\cite{JKKLSSV20}]
$\FC \leq \EC \leq \C \leq O(\Rzero) \leq O(\EC^2).$ 
\end{proposition}

 \subsection{Sensitivity and its variants }

Sensitivity is a lower bound on most of the measures described above (except $\Q$ and $\pAdv$). Given a Boolean function $f$, an input $x$ is \emph{sensitive at index $i$} if flipping the bit at index $i$ (which we denote by $x^{(i)}$) changes the value of the function to $1{-}f(x)$. 

\begin{definition}[Sensitivity]
\label{def:s}
$\s(f;x)$ is the number of sensitive indices of $x$. $\s(f) = \max_x \s(f;x)$
\end{definition}

If $B$ is a subset of indices, an input $x$ is \emph{sensitive to block $B$}  if simultaneously flipping all the bits in $B$ (which we denote by $x^B$)  changes the value of the function to $1{-}f(x)$.

\begin{definition}[Block sensitivity]
\label{def:bs}
$\bs(f;x)$ is the maximum number of disjoint sensitive blocks of $x$. $\bs(f) = \max_x \bs(f;x)$
\end{definition}

Aaronson et al.~\cite{ABKRT21} recently revived interest in a measure called $\lambda$. It was first introduced by Koutsoupias~\cite{K93}, and is a spectral relaxation of sensitivity.
\begin{definition}[Spectral sensitivity, or $\lambda$]
\label{def:lambda}
For a Boolean function $f$, let $F$ be the $|\Dom{f}| \times |\Dom{f}|$ 
matrix defined by $F(x,y)=1$ when $f(x)= 1-f(y)$ and $x,y$ differ in 1 bit. Then $\lambda(f)= \norm{F}$, where $\norm{{\cdot}}$ is the spectral norm.
\end{definition}

Note that $F$ can also be taken to be a $|f^{-1}(0)| \times |f^{-1}(1)|$ matrix with rows indexed by elements of  $f^{-1}(0)$ and columns by elements of $f^{-1}(1)$. It is easy to show that both ways of defining $F$ give the same spectral norm.

\begin{proposition}[\cite{ABKRT21,Tal13,GSS16,LLS06} ]
For any (possibly partial) Boolean function~$f$,
\[\lambda(f) \leq \s(f) \leq \bs(f) \leq \FC(f) \text{ and } \lambda(f)\leq \pAdv(f)\]
\end{proposition}

\subsection{Additional definitions for partial functions}
\label{sec:partialdefs}

Extending the definition of certificates to partial functions is slightly complex.
For $f : \{0,1 \}^n \rightarrow \{0,1,*\}$ it is natural to define the measures $\C^0(f)$, $\C^1(f)$, as well  as
$\C^{\{0,*\}}(f)$ and $\C^{\{1,*\}}(f)$ as follows:
\begin{definition}
\label{def:Cpartial1}
For $f : \{0,1 \}^n \rightarrow \{0,1,*\}$  and $b \in \{0,1 \}$ a partial assignment $\alpha$ is a $b$-certificate for $x \in f^{-1}(b)$
if $\alpha$ is consistent with $x$, and for any $x'$ consistent with $\alpha$ $f(x') = b$.

For $f : \{0,1 \}^n \rightarrow \{0,1,*\}$  and $b \in \{0,1 \}$ a partial assignment $\alpha$ is a $\{b,*\}$-certificate for $x \in f^{-1}(b)$
if $\alpha$ is consistent with $x$, and for any $x'$ consistent with $\alpha$ $f(x') \in \{b,*\}$.

For $b \in \{0,1\}$ and $x \in f^{-1}(b)$, 
$\C^{b}(f;x)$ is the size of the smallest  $b$-certificate for $x$ and
$\C^{b}(f) = \max_{x\in f^{-1}(b)} \{\C^{b}(f;x)\}$.  

For $b \in \{0,1\}$ and $x \in f^{-1}(b)$, 
$\C^{\{b,*\}}(f;x)$ is the size of the smallest  $\{b,*\}$-certificate for $x$ and
$\C^{\{b,*\}}(f) = \max_{x\in f^{-1}(b)} \{\C^{\{b,*\}}(f;x)\}$.
\end{definition}

Note that for example, while one can think of $0$-certificates for $x$ certifying that $f(x)=0$,
a $\{0,*\}$-certificate for $x$ certifies that $f(x) \neq 1$.
We also note that in the definition of $\C^{\{b,*\}}(f)$ we take the maximum over $x\in f^{-1}(b)$,
we do not include inputs $x$ where the function is not defined (e.g. where $f(x)=*$).

The above definitions are fairly straightforward and natural, but it is not immediately clear how to define $\C(f)$
for partial functions. We use the following notation:
\begin{definition}
\label{def:Cpartial2}
For $f : \{0,1 \}^n \rightarrow \{0,1,*\}$ we define
$\C(f) = \max \{\C^{\{0,*\}}(f), \C^{\{1,*\}}(f)\}$ and  
$\C'(f) = \max \{\C^0(f), \C^1(f)\}$.
\end{definition}

Notice that $\C(f) \leq \C'(f)$ for any $f$,  and for total functions $\C(f) = \C'(f)$. However, for partial functions $\C(f)$  can be much 
smaller than $\C'(f)$. The "Greater than Half" function (see section~\ref{sec:FunctionExample}) is an example of a partial function on $n$ bits with
$\C(f)= O(1)$ while $\C'(f) = \Theta(n)$.

It turns out that some results known for total functions remain valid for partial functions with respect to
$\C(f)$ but not with respect to $\C'(f)$ and others remain valid for partial functions with respect to
$\C'(f)$ but not with respect to $\C(f)$.
Thus, it is important to distinguish between the two versions.
We prefer to use this definition for $\C(f)$ since for example with this definition $\C(f)$ remains a lower bound
on deterministic query complexity (and for $\Rzero$ as well)  for partial functions. On the other hand, it is easy to construct partial functions with deterministic query complexity $O(1)$
but $\C'(f) = \Omega(n)$.
Some of our results for total functions involving $\C(f)$ 
no longer hold for partial functions, even though they remain valid with respect to $\C'(f)$.

A property of certificates often exploited in proofs is that every $0$-certificate must intersect (and contradict) every $1$-certificate
and this remains the case for partial functions. However, this property no longer holds for $\{0,*\}$ versus $\{1,*\}$-certificates.
Proofs based on this property remain valid for partial functions with respect to $\C'(f)$, but may no longer hold
for partial functions with respect to $\C(f)$.
An important example where this happens is the result that
$\EC(f)  \leq \C(f)$ by~\cite{JKKLSSV20}. This result does not hold for partial functions, as shown by the 
``Greater than Half'' function which has $\C(f) = O(1)$ and $\EC(f) = \Theta(n)$ (see section~\ref{sec:FunctionExample}), but remains valid with respect to $\C'(f)$.

For sensitivity  (block sensitivity) of partial functions, we consider an input $x$ in the domain $f^{-1}(0) \cup f^{-1}(1)$ to be sensitive to an index (or to a block)  if flipping it gives an input where $f$ is defined and takes the complementary value $1 - f(x)$. We do not consider an input to be sensitive to an index (or block) if
flipping it gives an input where $f$ is undefined.
Notice that with our definition, sensitivity can be 0 even for non-constant partial functions. 

\section{Public and private randomness in certificate games}



As a starting point, we give an upper bound of $\C$ on $\CGpub$ using a public coin protocol which illustrates how shared randomness can be used by the players to coordinate their outputs (Section~\ref{sec:CGpub-C-EC}). We then go on to show $\EC$ (Section~\ref{sec:CGpub-C-EC}), $\R$ and $\RS$ (Section~\ref{sec:CGpub-R}) are upper bounds on $\CGpub$. Finally, we give several upper bounds on private coin variant, $\CG$ (Section~\ref{sec:CG-bounds}).

\subsection{Public coin certificate game for the Tribes function }
\label{sec:CGpub-Tribes}

The $\Tribes_{s,t}$ function is a composition of two functions, $\Tribes_{s,t} = \OR_{s}\circ \AND_{t}$.

\begin{definition}[$\Tribes$]
\label{def:tribes}
$\Tribes_{s, t}:\{0,1\}^{st}\to \{0,1\}$ is defined using the DNF formula
\[\Tribes_{s, t}(x)=\bigvee_{i=1}^{s}\bigwedge_{j=1}^{t}x_{i,j}.\]
\end{definition}

The $\Tribes$ function is a very well studied problem in complexity theory. It has full randomized query complexity, in particular, $\R(\Tribes_{\sqrt{n}, \sqrt{n}}) = \Theta(n)$.   On the other hand, the functions $\OR_s$ and $\AND_t$ have full sensitivity. Thus $\CGpub$ of $\OR_{\sqrt{n}}$ and $\AND_{\sqrt{n}}$ is $\Theta(\sqrt{n})$. As a warmup to our general upper bounds on $\CGpub$, in Theorem~\ref{thm:tribes} we give a direct proof that the $\CGpub$ of $\Tribes_{\sqrt{n}, \sqrt{n}}$ is $O(\sqrt{n})$.
(This also follows from Theorem~\ref{thm:cgpub-is-fc}, the fact that for total functions $\CMM(f)=O(\C(f))$, and the fact that $\C(\Tribes_{\sqrt{n}, \sqrt{n}}) = \sqrt{n}$.)
Thus the function $\Tribes_{\sqrt{n}, \sqrt{n}}$ demonstrates a quadratic separation between $\R(f)$ and $\CGpub(f)$.

\begin{theorem}
\label{thm:tribes}$\CGpub(\Tribes_{\sqrt{n}, \sqrt{n}}) = 
O(\sqrt{n}).$
\end{theorem}

\begin{proof}

We give a public coin strategy for the Certificate game.
Let $x$ and $y$ be the two strings given to Alice and Bob respectively, that is $\Tribes_{\sqrt{n}, \sqrt{n}}(x) = 0$ and $\Tribes_{\sqrt{n}, \sqrt{n}}(y) =1$.  

Since $\Tribes_{\sqrt{n}, \sqrt{n}}(x) = 0$ for all $1\leq i\leq \sqrt{n}$ there exists $a_i$ such that $x_{i, a_i} = 0$ { where $x_{i, a_i}$ denotes the $a_i$th bit of the $i$th block of $x$}. Note that the $a_i$ is not necessarily unique. For each $i$, Alice arbitrarily picks a $a_i$ such that   $x_{i, a_i} = 0$ and then Alice considers a new string $x'$ where for all $i$,  $x_{(i,a_i)} = 0$ and for other bits of $x'$ is $1$.

Similarly, $\Tribes_{\sqrt{n}, \sqrt{n}}(y) = 1$ implies there exists an $b$ such that for all $1\leq j\leq \sqrt{n}$, $y_{b,j} = 1$. Again, note that there might be multiple such $b$ but Bob picks one such $b$ and considers the input $y'$ where $y_{b,j} = 1$ for all $1\leq j\leq \sqrt{n}$ and all other bits of $y'$ is set to $0$.

Note that $(b, a_b)$ is the unique index $(i,j)$ such that $x'{(i,j)}=0$ and $y'(i,j) = 1$. We will now present a protocol for Alice and Bob for outputting the index $(b, a_b)$ with probability at least $1/\sqrt{n}$. Note that this would imply our theorem. 

\begin{tcolorbox}
\begin{itemize}
    \item Alice and Bob uses shared randomness to select the same list of $\sqrt{n}$ permutations $\sigma_1, \dots, \sigma_{\sqrt{n}}: [\sqrt{n}]\to [\sqrt{n}]$, where the permutations are drawn (with replacement) uniformly and independently at random from the set of all possible permutations from $[\sqrt{n}]$ to $[\sqrt{n}]$.
    \item According to their pre-decided strategy both Alice and Bob picks the  same index $t$  between $1$ and $\sqrt{n}$.
    \item  Bob outputs $(b,\sigma_b^{-1}(t))$.
    \item Alice picks a number $i$ such that  $\sigma_i(a_i) = t$ and outputs $(i, a_i)$. In case no $i$ exists then Alice outputs any random index.
\end{itemize}  
\end{tcolorbox}

The probability of success of the protocol crucially depends on the fact that because Alice and Bob has shared randomness, they can pick the same set of  permutations $\sigma_1, \dots, \sigma_{\sqrt{n}}$ although the permutations are picked uniformly at random. 

We will show that with constant probability there exists a unique $i$ which satisfies $\sigma_i(a_i) = t$. Under the condition that this holds we will show that the probability of success of the above protocol is at least $1/\sqrt{n}$ which would prove the theorem. 
We start with the following claim that we will prove later.

\begin{claim}\label{cl:tribeshash}
For any fixed number $t$,  with probability at least  $(1-1/\sqrt{n})^{\sqrt{n}-1} \approx e^{-1}$,
there exists a unique $i$ such that $\sigma_i(a_i) = t$.
\end{claim}

Note that the permutation $\sigma_b$ is picked from the uniform distribution over all possible permutations from $[\sqrt{n}]$ to $[\sqrt{n}]$, i.e. $\sigma_b$ is a random bijection from $[\sqrt{n}]$ to $[\sqrt{n}]$. So with probability $1/\sqrt{n}$, $t = \sigma_b(a_b)$. Assuming that $t = \sigma_b(a_b)$ and that there exists a unique $i$ such that $\sigma_i(a_i) = t$, note that output of both Alice and Bob is indeed $(b, a_b)$. Thus the probability of success of the protocol is $\Omega(1/\sqrt{n})$.
\end{proof}

\begin{proof}[Proof of Claim~\ref{cl:tribeshash}] Consider the event $$\mathcal{E}_k:= \mbox{$\sigma_k(a_k) = t$ and for all $i\neq k$, $\sigma_i(a_i) \neq t$}.$$ The 
 probability that the event $\mathcal{E}_k$ occurs is $\frac{1}{\sqrt{n}}\cdot (1-\frac{1}{\sqrt{n}})^{\sqrt{n} - 1}$. The event that there exists a unique $i$ such that $\sigma_i(a_i) = t$ is $\cup_{k=1}^{\sqrt{n}} \mathcal{E}_k$. The events $\mathcal{E}_k$ are disjoint and the claim follows.
\end{proof}



\subsection{Upper bounds on \texorpdfstring{$\CGpub$}{CGpub} by \texorpdfstring{$\C$}{C} and \texorpdfstring{$\EC$}{EC}}
\label{sec:CGpub-C-EC}

We will take advantage of having access to shared randomness by using the hashing based approach outlined in Section 
\ref{sec:overview-upper}.
To illustrate the ideas of the proof,
we start with a simple argument to show that $\CGpub$ is always
upper bounded by certificate complexity.

Both players pick a certificate for their respective inputs.
They permute the indices $\{1,\ldots,n\}$  with shared randomness and each player outputs the first index in this new order within their certificate. Since their certificates must intersect, the probability that they are correct is at least one over the size of the union which is at most $2\C(f)$, so for any total function $
\CGpub(f) \leq 2\C(f))$.\footnote{We thank an anonymous referee for suggesting this simple and elegant proof.}

We now provide a different proof based on hashing as a warmup before we prove the stronger result $\CGpub(f)= O(\EC(f))$.

\begin{theorem}\label{thm:CGpub-C}
  For a total Boolean function $f$, $\CGpub(f) \leq O(\C(f))$.
\end{theorem}

\begin{proof}
  Let $S$ be a finite set of cardinality $\C(f)$.
  An element  $z \in S$ is fixed as part of the specification of the
  protocol ($z$ does not depend on the input).
  
  Using shared randomness, the players select a function
  $h: [n] \rightarrow S$ 
  as follows.
    Let $h: [n] \rightarrow S$ be a random
  hash function such that for each $i \in [n]$, $h(i)$ is selected
  independently and uniformly from $S$.
  
  For $x \in f^{-1}(0)$ we fix an optimal 0-certificate $C_x$,
  and denote by $A_x \subseteq [n]$ the set of indices fixed by $C_x$.
Similarly, for $y \in f^{-1}(1)$ we fix an optimal 1-certificate $C_y$,
and denote by $B_y \subseteq [n]$ the set of indices fixed by $C_y$.

After selecting $h$ using shared randomness, the players
proceed as follows.
On input $x$, Alice outputs an index $i \in A_x$ such that $h(i) = z$,
and on input $y$, Bob outputs an index $j \in B_y$ such that $h(j) = z$.
If they have several valid choices, they select randomly, and if they have
no valid choices they output arbitrary indices.

Let $i^* \in A_x \cap B_y$, such that $x_{i^*} \neq y_{i^*}$.
By the definition of certificates, such element $i^*$ exists for  any
$x \in f^{-1}(0)$ and $y \in f^{-1}(1)$, and $i^*$ is a correct answer
on input $(x,y)$ if both players output~$i^*$.
Next, we estimate what is the probability that both players output $i^*$.

First recall that by the definition of $h$,
the probability that $h(i^*)=z$
is $\frac{1}{|S|} = \frac{1}{\C(f)}$.
Next, notice that for any $i \in A_x \cup B_y$
the number of elements different from $i$ in $A_x \cup B_y$ is
$\ell = |A_x \cup B_y| -1 
\leq |A_x| + |B_y| - 2$.
Thus
for any $z \in S$ and any $i \in A_x \cup B_y$
the probability (over the choice of $h$)
that no element other than $i$
in $A_x \cup B_y$
is mapped to $z$ by~$h$ is
$( 1 - \frac{1}{|S|})^{\ell} \geq \frac{1}{e^2}$, since
$\max \{|A_x|, |B_y| \} \leq \C(f) = |S|$ and thus $\ell \leq 2(|S|-1)$.

Thus, the players output a correct answer with probability at least
$\frac{1}{e^2} \frac{1}{\C(f)}$.
\end{proof}

The previous theorem is stated for total functions and its proof critically depends on the intersection property of 0- and 1-certificates which does not hold for $\{0,*\}$- vs. $\{1,*\}$-certificates.
The theorem fails to hold for the partial function "Greater than Half" (see Section~\ref{sec:FunctionExample}), for which it is the case that $\C(\GTH)=1$ whereas $\CGpub(\GTH)$ is $\Theta(n)$.
However, the theorem and its proof remain valid for partial functions with respect to $\C'(f)$ (see Section~\ref{sec:partialdefs}). We obtain a stronger upper bound on $\CGpub$ by $\EC$.

\begin{theorem}\label{thm:CGpub-EC}
For a (possibly partial) Boolean function~$f$, $\CGpub(f) \leq O(\EC(f))$.
\end{theorem}

\begin{proof}
  The proof will be similar but slightly more involved than the
  proof of the upper bound by $\C$.
  We will rely on the ``weights'' $w_{x,i}$ from the definition of $\EC(f)$.

  Let $S$ be a finite set of cardinality $\lceil \EC(f) \rceil$.
  Using shared randomness, the players select a function
  $h: [n] \rightarrow S$ and an element $z \in S$ as follows.
    Let $h: [n] \rightarrow S$ be a random
  hash function such that for each $i \in [n]$, $h(i)$ is selected
  independently and uniformly from $S$. In addition,
  $z$ is selected uniformly from $S$
  and independently from the choices for $h$.

  For all inputs $x \in \{0,1\}^n$ consider the weights $w_{x,i}$
  achieving $\EC(f)$. 
  Denote by  $EC_x$ the sum $\sum_{i \in [n]}w_{x,i}$
  and recall that by the definition of $\EC$,
  for each $x \in \{0,1\}^n$ we have
$EC_x \leq \EC(f)$.
  
    For a given $z \in S$, consider the preimage $h^{-1}(z)$.
  We use the notation
  $$W_x(z) =  \sum_{i \in h^{-1}(z)} w_{x,i}\;.$$
  Notice that for any $z \in S$,
  $$E[W_x(z)] = \sum_{i \in [n]}\frac{w_{x,i}}{|S|} = \frac{EC_x}{|S|} $$
  where the expectation is over the choice of the hash function.

  After selecting $h$ and $z$ using shared randomness, the players
proceed as follows.
On input $x \in f^{-1}(0)$ Alice selects an index $i$ from $h^{-1}(z)$
such that each $i$ is chosen with probability $\frac{w_{x,i}}{W_x(z)}$.
Similarly, on input $y \in f^{-1}(1)$ Bob selects an index $i$ from $h^{-1}(z)$
such that each $i$ is chosen with probability $\frac{w_{y,i}}{W_y(z)}$.
Note that these choices are made using Alice's and Bob's private randomness,
so for fixed $z$ and $h$ Alice's choices are independent from Bob's choices.
However,
they both depend on $z$ and $h$.
In what follows, we will denote by $Pr_z$ and $Pr_h$, respectively,
the probabilities that are only over the choice of $z$ and $h$, respectively.

Recall that $W_x(z)$ and $W_y(z)$ are measures of the preimage of $z$
with respect to the weights for $x$ and $y$ respectively. 
Since $\frac{\EC_x}{|S|} \leq 1$ for any $x \in \{0,1\}^n$,
the preimage of most elements in $S$ will have small measure.
  Next we estimate the probability that a given element $i$ is mapped
to a value $h(i)$ whose preimage has small measures $W_x(h(i))$ and
$W_y(h(i))$. Note that this only depends on the choice of $h$.

  For a given $i$, consider first selecting the values $h(j)$
for all $j \neq i$ from $[n]$.
Consider the measure of the preimages of elements in $S$ at this point
(without taking into account what happens to $i$).
Since $\frac{\EC_x - w_{x,i}}{|S|} \leq 1$ for any $x \in \{0,1\}^n$,
at most $\frac{1}{t-1}$ fraction of the elements in $S$ can have measure
more than $t-1$ at this point.
Since $w_{x,i} \leq 1$, we get that for any $x \in \{0,1\}^n$ and $i \in S$,
$Pr_h[ W_x(h(i)) > t] \leq \frac{1}{t-1}$.

For $i \in [n]$,
let $\SMALL_i$ denote the event that both $W_x(h(i))$ and $W_y(h(i))$
are at most $t$.
Then $Pr_h[\SMALL_i] \geq 1 - \frac{2}{t-1}$.

For a given $i \in [n]$,
let $\BOTH_i$ denote the event that both players select
$i$.  
Let $I(x,y) = \{i | x_i \neq y_i \}$. Since $f(x) = 1- f(y)$,
$I(x,y) \neq \emptyset $.

Recall that the players goal is that they both output the same $i$
from $I(x,y)$.
Denote by $P(x,y)$ the probability
that they both output the same $i$ from $I(x,y)$. 
Note that $P(x,y)$ is at least as large as the
probability that they both output the same $i$
from $I(x,y)$, and both $W_x(h(i))$ and $W_y(h(i))$
are at most $t$.

Thus, using that the events $\BOTH_i$ are pairwise disjoint, we have

\begin{equation*}
  \begin{split}
P(x,y) & \geq 
\sum_{i \in I(x,y)} Pr[ \BOTH_i \cap (z{=}h(i)) \cap \SMALL_i] \\
& = \sum_{i \in I(x,y)}
Pr[ \BOTH_i | (z{=}h(i)) \cap \SMALL_i] Pr[(z{=}h(i)) \cap \SMALL_i] \;.
  \end{split}
\end{equation*}

Note that the events $z=h(i)$ and $\SMALL_i$ are independent,
since the choice of $z$ is independent of $h$.
For any $i^* \in I(x,y)$, and $h: [n] \rightarrow S$,
$Pr_z[z = h(i^*)]= \frac{1}{|S|}$,
Thus,
$Pr[z=h(i) \cap \SMALL_i] = Pr_z[z=h(i)] Pr_h[\SMALL_i]
= \frac{1}{|S|}Pr_h[\SMALL_i] \geq \frac{1}{|S|} (1 - \frac{2}{t-1})$.

For any $i \in [n]$, we have
$$Pr[ \BOTH_i | z=h(i)] = \frac{w_{x,i}}{W_x(z)} \frac{w_{y,i}}{W_y(z)}
\mbox{  and  }
Pr[ \BOTH_i | z=h(i) \cap \SMALL_i] \geq \frac{w_{x,i}}{t} \frac{w_{y,i}}{t}\;.$$
Thus, we get
$$P(x,y) \geq \frac{1}{t^2} \frac{1}{|S|} (1 - \frac{2}{t-1})
\sum_{i \in I(x,y)} w_{x,i} w_{y,i}
\geq \frac{1}{t^2} \frac{1}{|S|} (1 - \frac{2}{t-1})$$
where the last inequality follows by the definition of $\EC(f)$.

Setting $t = 5$, we get that
the players output the same element from $I(x,y)$ with probability
at least
$\frac{1}{50} \frac{1}{\lceil \EC(f) \rceil} = \Omega(\frac{1}{\EC(f)})$.
\end{proof} 
   

\subsection{Lower bound on sabotage complexity}
\label{sec:CGpub-R}

Randomized sabotage complexity $\RS$~\cite{BK18} is a measure of complexity introduced to study the behavior of randomized query complexity $\R$ under composition. It was shown that $\RS$ is a lower bound on $\R$ and that it behaves perfectly under composition.
\begin{definition}[Sabotage Complexity \cite{BK18}]
\label{def:sabotage}
The sabotage complexity of a function $f$, denoted $\RS(f)$, is defined using a concept of \emph{sabotaged}
inputs $P_f \subseteq \{0,1,*\}^n$ which is the set of all partial assignments of a function $f$ consistent with a $0-$input and a $1-$input. Let $P_f^\dagger$ is defined similarly with the symbol $*$ being replaced by $\dagger$. Given a (possibly partial) function $f$, a partial function $f_{sab} : P_f \cup P_f^\dagger \mapsto \zone$ is defined as $f_{sab}(x) = 1 $ if $x \in P_f$ and $f_{sab}(x) = 0 $ if $x \in P_f^\dagger$ (here we view $P_f,P_f^\dagger$ as subsets of $\{0,1,*,\dagger\}^n$). The sabotage complexity is defined as the randomized query complexity of computing $f_{sab}$ i.e. $\RS(f) = \R(f_{sab})$.
\end{definition}

The classical adversary method $\CMM$ was introduced as a lower bound on $\R$~\cite{LM08} but there were no limitations known on this quantity that hold for partial functions~\cite{AKPV}.
In this section we  show that on sabotage complexity $\RS$ is an upper bound on $\CGpub$ and therefore on $\CMM$ (see Theorem~\ref{thm:cgpub-is-fc}). As a warm-up, we give an easy proof that 
$\CGpub(f) = O(\R(f))$. 

\begin{theorem}
\label{prop:CGpub-R}
For any Boolean (possibly partial) function $f$,
$\CGpub(f) \leq O(\R(f)).$
\end{theorem}

\begin{proof} 
From the definition of $\R(f)$ there is a randomized decision tree $\mathcal{R}$ that on any input $x$ outputs $f(x)$ correctly with probability at least $2/3$, and  $\mathcal{R}$ only reads at most $\R(f)$ number of bits of $x$. To prove $\CGpub(f) \leq \R(f)$ let us consider the following strategies used by the two players: 

\textit{Both the players run the algorithm $\mathcal{R}$ on their respective inputs using the same random coins (using the shared randomness). Both the player also use shared randomness to  pick a number $t$ uniformly at random between $1$ and $\R(f)$. Both the players output the $t$-th index that is queried by~$\mathcal{R}$.}

Let $x$ and $y$ be the inputs to the players respectively.
Since $f(x) = 1- f(y)$, with probability at least $4/9$ the algorithm $\mathcal{R}$
will output different answers when the players run the algorithm on their respective inputs. Also since the algorithm $\mathcal{R}$ is run using the same internal coins, the initial sequence of indices queried by both the runs of the algorithm is the same until the algorithm queries an index $k$ such that $x_k \neq y_k$. Note that with probability $1/\R(f)$, the random number $t$ picked by $t$ is the same as $k$. So with probability $\frac{4}{9}\cdot \frac{1}{\R(f)}$, the players correctly output the same index $t$ such that $x_t \neq y_t$. Hence $\CGpub(f) \leq O(\R(f))$.
\end{proof}

Using the same idea we can show that the public coin certificate game complexity $\CGpub$ is bounded above by  randomized sabotage complexity.

\begin{theorem}
\label{thm:sabotage}
The public coin certificate game complexity of a (possibly partial) function~$f$ is at most its sabotage complexity: 
$\CGpub (f) \leq \frac{9}{2} \RS(f)$.
\end{theorem}
\begin{proof}
We show this by using the sabotage complexity protocol to build a $\CGpub$ protocol. 
Assuming that Alice has input $x$ and Bob an input $y$ such that $f(x)= 1- f(y)$, we construct a sabotaged input $z_{x,y}$ that is consistent with $x$ and $y$ as follows:
\begin{align*}
z_{x,y} (i) = 
\begin{cases}
x(i) &\quad\text{ if } x(i) = y(i)\\
* &\quad\text{ otherwise.}
\end{cases}
\end{align*} 
On the input $z_{x,y}$, we know that a decision tree sampled from the distribution given by the $\RS$ protocol succeeds in finding a $*$ or $\dagger$ with probability $\geq 2/3$. The $\CGpub$ protocol is as follows: using public randomness, Alice and Bob sample a decision tree from the $\RS$ protocol and follow the path on the decision tree according to their respective inputs for at most $\RS(f)$ steps.
With probability at least $2/3$ the randomly chosen tree finds a $*$ on input $z_{x,y}$ in $\RS(f)$ steps. Since the sabotaged input $z_{x,y}$ is consistent with both Alice's and Bob's input, the path on $x$ and $y$ on the decision tree is the same as that on $z_{x,y}$ until they reach a place where they differ (or encounter a $*$ in $z_{x,y}$). Alice and Bob pick a random position $t$ such that $1 \leq t \leq \RS(f)$ and output the $t^{th}$ query made in the corresponding paths on the tree. With probability $\frac{1}{\RS(f)}$, it is place corresponding to a $* \in z_{x,y}$ and they succeed in finding a place where the inputs differ. This gives a success probability $\geq 2/3 \frac{1}{\RS(f)}$ as the random decision tree sampled finds a $*$ on the sabotaged input $z_{x,y}$ with probability $\geq 2/3$.    
\end{proof} 



\subsection{Upper and lower bounds for private coin certificate games} \label{sec:CG-bounds}

We first observe that the following formulation is equivalent to ${\CG}$. The essential idea is rescaling, and the objective function gets squared because the constraints are quadratic.
\begin{proposition}[Equivalent formulation for $\CG$]\label{prop:sqrt-CG}
For any (possibly partial) function $f$,
 \begin{align*}
 {\CG(f)}  = &  \min_{\{w_{x,i}\}} \max_x  \left\{\sum_i w_{x,i}\right\}^2
  \\
  \text{such that } &  
  \sum_{i: x_i \neq y_i}  w_{x,i} w_{y,i} \geq 1 \quad \forall x\in f^{-1}(0), y\in f^{-1}(1) \\
  ~& w_{x,i} \geq 0 \quad \forall x,i
 \end{align*}
\end{proposition}

\begin{proof}
We will first show that the value of the objective function in the formulation in terms of weights is at most ${\CG}$. 
Let $p$ be an optimal probability distribution that achieves $\CG(f)$ and let 
$$\Delta = \min_{x, y: f(x) = 1 - f(y)}  \sum_{i: x_i \ne y_i} {p_{x,i} p_{y,i}} = \frac{1}{\CG(f)}. $$
We construct the following weight scheme using $p$,  
$w_{x,i} = \frac{p_{x,i}}{\sqrt{\Delta}}$ 
and this is a feasible solution for the above formulation since $\forall  x, y$ such that $ f(x) = 1- f(y)$,
\begin{align*}
   \sum_{i: x_i \ne y_i} {w_{x,i} w_{y,i}} = \frac{1}{\Delta} \sum_{i: x_i \ne y_i} {p_{x,i} p_{y,i}} \geq \frac{\Delta}{\Delta} =1
\end{align*}
We now have
\begin{align*}
 \min_{\{w'_{x,i}\}} \max_{x} \left\{\sum_{i \in [n]} w'_{x,i}\right\}^2 \leq \max_{x}  \left\{\sum_{i \in [n]} w_{x,i}\right\}^2 = \max_{x} \left\{\sum_{i \in [n]}  \frac{p_{x,i}}{\sqrt{\Delta}}\right\}^2 =  \frac{1}{{\Delta}} = {\CG(f)}
\end{align*}
For the other direction, let $w$ be an optimal weight scheme $w$ that minimises $\max_x  \sum_i w_{x,i}$. We construct the following family of probability distributions:
$p_{x,i} = \frac{w_{x,i}}{\sum_j w_{x,j}}$ 
This gives the following.
 \begin{align*}
\CG(f) \leq &
\max_{\substack{x, y\\ f(x) = 1- f(y)}} \frac{1}{\sum_{i: x_i \ne y_i} {p_{x,i} p_{y,i}}} \\
= & \max_{\substack{x, y\\ f(x) = 1- f(y)}} \frac{\sum_j w_{x,j} \sum_j w_{y,j}}{\sum_{i: x_i \ne y_i} w_{x,i} w_{y,i}} \\
\leq & \max_{\substack{x, y\\ f(x) = 1- f(y)}} \sum_j w_{x,j} \sum_j w_{y,j}.
\end{align*}
Thus we have 
${\CG(f)} \leq \max_{x} \left\{ \sum_j w_{x,j} \right\}^2$.
\end{proof}

We show that the following  relations hold for $\CG$.
\begin{theorem}
\label{thm:CG-bounds}
For any total Boolean function $f$, 
\begin{multicols}{2}
\begin{enumerate}
    \item  $\pAdv(f)^2 \leq \CG(f) $
    \item $\R_0(f) \leq \CG(f)\leq O(\EC(f)^2)$ \cite{JKKLSSV20}
    \item $\CG(f) \leq O(\CGpub(f)^2 \s(f))$ \cite{JKKLSSV20}
    \item $\CG(f) \leq \C^0(f)\C^1(f)$
\end{enumerate}
\end{multicols}

\end{theorem}
The first and last items also hold for partial functions.
However,  
the ''Greater than Half" function (see section~\ref{sec:FunctionExample}) is an example of a partial function that would violate item 4 using the alternate definitions $\C^{\{0,*\}}$ and $\C^{\{1,*\}}$.

\begin{proof}

{\bf Item 1} Let $p$ be an optimal solution for $\CG(f)$ so that 
$\omega(p;x,y)\geq \frac{1}{\CG(f)} $ for all $x,y$ satisfying $f(x)= 1- f(y)$.
Using the same assignment for $\pAdv$ (Definition~\ref{def:pAdv}), 
it is the case that
\begin{align*}
\frac{1}{\pAdv(f)^2}&\geq 
\min_{\substack{x \in f^{-1}(0)\\ y \in f^{-1}(1)}} \left(\sum_{i:x_i \neq y_i}\sqrt{p_{x,i}p_{y,i}}\right)^2\\[3pt]
&\geq \min_{\substack{x \in f^{-1}(0)\\ y \in f^{-1}(1)}} \sum_{i:x_i \neq y_i}{p_{x,i}p_{y,i}}
\end{align*} 
so $\pAdv(f)^2 \leq \CG(f) $.

{\bf Item 2} From Proposition~\ref{prop:sqrt-CG}, the formulation of $\sqrt{CG}$ is a relaxation of the definition of $\EC$, where the constraint $w_{x,i} \leq 1$ is dropped in 
$\sqrt{CG}$,
giving the second inequality $\sqrt{CG}(f) \leq \EC(f)$.

For the first inequality, it was shown in~\cite{JKKLSSV20} that $\R_0\leq O(\EC^2)$. However, their proof does not make use of the constraints $w_{x,i} \leq 1$. 
Therefore, their proof already shows that $\R_0(f) \leq O(\CG(f))$.

{\bf Item 3} Jain et al.~\cite{JKKLSSV20} showed that $\EC(f)^2\leq O(\FC(f)^2 \s(f))$. From the previous item $\CG(f)\leq O(\EC(f)^2)$, and $\FC(f) = O(\CMM(f)) 
$~\cite{AKPV}, $\CMM(f)\leq \CGns(f) \leq \CGpub(f)$ from Theorem~\ref{thm:CGns-cAdv} and Proposition~\ref{prop:CGhierarchy}. We get the desired result by combining these inequalities.

{\bf Item 4} It is easy to see that $\CG(f) \leq \C^0(f) \cdot \C^1(f)$:
on input $x$, each player outputs uniformly at random some 
index $i$ in a minimal certificate for their input. The certificates must intersect in at least one index, otherwise we could simultaneously fix the value of $f$ to 0 and to 1 by fixing both certificates. The strategy therefore succeeds
when both players output the same index in the intersection, which occurs with probability at least $\frac{1}{\C^0(f)}\frac{1}{\C^1(f)}$.
This argument remains valid for partial functions, however the ``Greater than Half'' function (see section~\ref{sec:FunctionExample}) is an example of a partial function that would violate item 4 using the alternate definitions $\C^{\{0,*\}}$ and $\C^{\{1,*\}}$.
\end{proof}


\section{Lower bounds on quantum certificate game complexity}
\label{sec:CGstar-lb}


In this section, we give a very short and simple proof that the classical adversary ($\CMM$)  is a lower bound on all of our certificate game models.

To illustrate the idea behind the proof  and the technique we use, we start with a quantum lower bound on the $\OR$ function. Consider a hypothetical  strategy with shared entanglement that would allow two players to win the certificate game with probability more than $1/n$. Then the players could use this strategy for the certificate game  as a black box, to convey information (without using communication) in the following way. Assume Alice wants to send an integer $i\in \{1,\ldots,n\}$ to Bob. Bob uses input $y=0^n$ and Alice uses input $x=y^{(i)}$ (all $0$s with the $i$-th bit 1). By running this game several times, Bob could learn $i$ by taking the majority output of several runs of this game, which would violate the non-signaling principle of quantum information.

In order to give a formal proof, we show that $\CGns(\pOR)\geq n$. 
Since $\CGns(f)\leq \CGstar(f)$ for every $f$, the following proposition implies that $\CGstar(\pOR_n)\geq n$.
\begin{proposition}
\label{prop:CGnsOR}
$\CGns(\pOR_n) \geq n$.
\end{proposition}

\begin{proof}[Proof of Proposition~\ref{prop:CGnsOR}]
We give a feasible solution to the dual, composed of a hard distribution $\mu$ and an assignment to the variables $\gamma_{i,j,x,y}$ that satisfy the constraints of the dual  given in Proposition~\ref{prop:CGns-dual}.

 Let $\delta=\frac 1 n, x = 0^n$, and consider $\mu_{xy}= \frac 1 n$ when $y = x^{(i)}$ ($x$ with the $i^{\text{th}}$ bit flipped to 1), and~0 everywhere else.
 To satisfy the correctness constraint, we use $\gamma$ to pick up weight $1/n$ whenever a strategy $AB$ fails on some pair $x,x^{(i)}$. To do this, we define $\gamma_{i,j,x,x^{(i)}} = \frac 1 n$ for all $j{\neq}i$ (and~0 everywhere else).
 To see that this satisfies the constraints, consider any strategy $AB$ and let $i = A(x)$ be $A$'s output on~$x$.

 \ 
 
 \textbf{Case 1:} If $B(x^{(i)})=i$ then $AB$ is correct on $x,x^{(i)}$, but cannot be correct  on any other input pair with non-zero weight under $\mu$. Therefore,
 $$\sum_{x,y': A(x) =B(y')=i  \text{ and } x_{i} \neq y_{i}}\mu_{x,y} = \frac 1 n \quad \text{and}\quad \sum_{x,y'}\gamma_{A(x), B(y'),x,y'} = 0.$$

 \ 
 
 \textbf{Case 2:} If $B(x^{(i)})=j\neq i$, then $AB$ is incorrect on all non-zero weight input pairs, and we have
 $$\sum_{x,y': A(x) =B(y')=i  \text{ and } x_{i} \neq y_{i}}\mu_{x,y}=0 \quad \text{ and} \quad  \sum_{x,y'}\gamma_{A(x),B(y'),x,y'}=\frac 1 n
 .$$ 
 
 Since $\delta=\frac 1 n$ this is satisfying assignment, which shows that $$\CGns(\pOR)= \winning^\ns(R_{\pOR})^{-1}\geq n .$$
 \end{proof} 
 
Note that $\Q(\pOR)$ is $\Theta(\sqrt{n})$~\cite{Grover96, BBC+01}. Thus, $\pOR$ shows that there exist a function for which $\CGstar(f) = \omega(\Q(f))$ (as opposed to the randomized model where $\CGpub(f) \leq O(\R(f))$. On the other hand,  
note that the function constructed by \cite{ABK16} demonstrates that there exists a total Boolean function $f$ with $\C(f) = O(\sqrt{\Q(f)})$; this $f$ also shows that $\CGpub(f)$ could be as small as $O(\sqrt{\Q(f)})$.

The previous lower bound on the $\OR$ function can be generalized,  with a slightly more complicated weight assignment, to show that block sensitivity is a lower bound on the non-signaling value of the certificate games. 
However, using a different technique, we can prove an even stronger result. We do this by going  back to the original definition of $\CGns$ (Definition~\ref{def:CGns}) and giving a very simple proof that $\CGns$ is an upper bound on $\CMM$.

\begin{theorem}
\label{thm:CGns-cAdv}
For any (possibly partial) Boolean function~$f$,  $\CMM(f)\leq \CGns(f)$.
\end{theorem}
\begin{proof}
Let $p(i,j|x,y)$ be the distribution over outcomes  in an optimal nonsignaling strategy for $\CGns(f)$.
Then $p$ verifies the nonsignaling condition, $\sum_j p(i,j|x,y)=\sum_j p(i,j|x,y')$ for all $x,y,y',i$, so we can write the marginal distribution for $x$ as $p(i|x)=\sum_j p(i,j|x,y)$, since it does not depend on $y$.
Notice that $p(i|x)=\sum_j p(i,j|x,y)\geq p(i,i|x,y)$ for all $x,y,i$,
so $\min \{p(i|x), p(i|y)\} \geq p(i,i|x,y)$.
\begin{align*}
\CMM(f) = & \min_p \max_{\substack{x, y\in S\\ f(x) = 1- f(y)}} \frac{1}{\sum_{i: x_i \ne y_i} \min\{p(i|x), p(i|y)\}} \\ 
\leq & \min_p \max_{\substack{x, y\in S\\ f(x) = 1- f(y)}} \frac{1}{\sum_{i: x_i \ne y_i} p(i,i|x,y)} 
\end{align*}
Since we have that  $\sum_{i:x_i\neq y_i} p(i,i|x,y)\geq \frac{1}{\CGns(f)}$ for all $x,y$ such that $f(x) = 1- f(y)$, $
\CMM(f) \leq {\CGns(f)}$. 
\end{proof}

To summarize the key idea  of this section, introducing the non-signaling model of Certificate games provides a very clean and simple way to give lower bounds on all of our previous models, including the shared entanglement model. It has several linear formulations, making it very easy to give upper and lower bounds. Finally, it captures an essential feature of zero-communication games, which we think of as the ``non-signaling bottleneck''.
As an added bonus, it allows us to give proofs on the shared entanglement model without having to get into the technicalities of what characterizes quantum games.


\section{Closing the loop}
\label{sec:cgpub-is-cmm}


In this section we will show that all of $\CGpub$, $\CGstar$, $\CGns$, and $\CMM$ are actually asymptotically equal.
\begin{theorem} \label{thm:cgpub-is-fc}
For any (possibly partial) Boolean function $f$,
\[
\CGpub(f)=\Theta(\CGstar(f))=\Theta(\CGns(f))=\Theta(\CMM(f)).
\]
\end{theorem}
The key idea is to apply the \textit{correlated sampling} technique of Holenstein~\cite{Hol09}. We use the following formulation from the Rao--Yehudayoff textbook~\cite[Lemma~7.5]{Rao20}. Here, \emph{total variation distance} between distributions $p$ and $q$ is defined by $\TV(p,q)\coloneqq \frac{1}{2}\sum_i|p(i)-q(i)|$.

\begin{lemma}[Correlated sampling~\cite{Hol09,Rao20}]
\label{lem:corr-sampling}
Suppose Alice is given as input a distribution $p$ over a set $\mathcal{U}$, and Bob is given as input a distribution $q$ over $\mathcal{U}$. There is a protocol using public randomness and no communication with the following guarantees.
\begin{itemize}
\item Alice outputs a value $X$ which is distributed according to $p$.
\item Bob outputs a value $Y$ which is distributed according to $q$.
\item We have\footnote{The original formulation from~\cite[Lemma~7.5]{Rao20} states the incomparable bound $\Pr[X\neq Y]\leq 2\TV(p,q)$ (which is useful when $\TV$ is small). However, it is straightforward to inspect the protocol and see that it also satisfies our lower bound (which is useful when $\TV$ is close to 1).}
$\Pr[X=Y]\geq \frac{1}{2}(1-\TV(p,q))$.
\end{itemize}
\end{lemma}

\begin{proof}[Proof of Theorem~\ref{thm:cgpub-is-fc}]
Given Theorem~\ref{thm:CGns-cAdv}, it remains to prove the inequality $\CGpub(f)\leq O(\CMM(f))$ by designing a protocol for $f$ that wins with probability $\Omega(1/\CMM(f))$. Recall from  Definition~\ref{def:CMM} that
\begin{equation} \label{eq:fc-cmm}
\CMM(f) = \min_p \max_{\substack{x \in f^{-1}(1) \\ y \in f^{-1}(0)}}  \frac{1}{\sum_{i:x_i \neq y_i} \min\{p_x(i),p_y(i)\}}.
\end{equation}
Starting with a distribution $p_x$ over $[n]$, we define distributions $p'_x$ and $p''_x$ over $[n]\times\{0,1\}$ as follows. To define $p'_x$, first sample $i \sim p_x$ and then output $(i,x_i) \sim p'_x$. To define $p''_x$, first sample $i \sim p_x$ and then output~$(i,1-x_i) \sim p''_x$. (Note how $p'_x$ and $p''_x$ are the same except for the flipped bit.) We can now write the denominator in \eqref{eq:fc-cmm} as
\begin{align*}
\sum_{i:x_i \neq y_i} \min\{p_x(i),p_y(i)\}
&=
\sum_{\alpha \in[n]\times\{0,1\}}\min\{p'_x(\alpha),p''_y(\alpha)\} \\
&=
\sum_\alpha{\textstyle \bigg(
\frac{1}{2}(p'_x(\alpha)+p''_y(\alpha)) - \frac{1}{2}|p'_x(\alpha)-p''_y(\alpha)|
\bigg)} \\
&=1 - \TV(p'_x,p''_y).
\end{align*}
The $\CGpub$ protocol for $f$ is now defined from the optimal $p$ in \eqref{eq:fc-cmm} as follows. On input $(x,y)$ the players use the protocol from Lemma~\ref{lem:corr-sampling}
to compute correlated samples $X\coloneqq (i_X,b_X)\sim p'_x$ and~$Y\coloneqq (i_Y,b_Y)\sim p''_y$, respectively, and then they output $(i_X,i_Y)$. The players win the game with probability
\[\textstyle
  \Pr[X=Y] \geq \frac{1}{2}(1-\TV(p'_x,p''_y))
  \geq \Omega(1/\CMM(f)). 
\]
\end{proof}

Combining the above theorem with  the result of~\cite{AKPV}, which states that $\FC(f)=\Theta(\CMM(f))$ for total~$f$, we get the following immediate corollary.

\begin{corollary}
For any total Boolean function $f$, $\CGpub(f)=\Theta(\FC(f))$.
\end{corollary}

\section{Single bit versions}

Aaronson et al.~\cite{ABKRT21} defined single-bit versions of
several formulations  of the  adversary method, and showed that 
they are all equal to the spectral sensitivity $\lambda$.
Informally, single-bit versions of these measures are obtained
by considering the
requirements only with respect to pairs $x,y$ such that 
$x,y\in f^{-1}(0) \times f^{-1}(1)$ and $x$ and $y$ differ only in a single bit.

We denote by $d(x,y)$ the Hamming distance of $x$ and $y$,
and by $x^{(i)}$ the string obtained from $x$ by flipping the value
of the $i$-th bit $x_i$ to its negation.
The single-bit version of $\MM(f)$ was defined in~\cite{ABKRT21}
as follows.
\begin{equation}\label{def:MM1one}
 \MM_{[1]}(f) = \min_{\{w_{x,i}\}} \max_x  \sum_i w_{x,i}
  \\
  \text{ such that }   
  w_{x,i} w_{x^{(i)},i} \geq 1 \quad \forall x,i \text{ with }
  f(x) = 1 - f(x^{(i)})
  \end{equation}
  where $x \in \{0,1\}^n$ and $i \in [n]$.

  Similarly to the proof of  Proposition \ref{prop:sqrt-CG}
  it can be shown that this is equal to the following formulation,
which we  include for comparison with some of our other definitions.
\begin{equation}\label{def:MM1two}
  \MM_{[1]}(f) := \min_p \max_{\stackrel{x,y\in f^{-1}(0) \times f^{-1}(1)}
  { d(x,y)=1}} \frac{1}{\sum_{i:x_i \neq y_i}\sqrt{p_{x,i}p_{y,i}}}
= \min_p \max_{x,i: f(x)=1- f(x^{(i)})} \frac{1}{\sqrt{p_{x,i}p_{x^{(i)},i}}}
\end{equation}
where $p$ is taken over all families of
nonnegative $p_{x,i} \in \mathbb{R}$ such that for all $x$, $\sum_{i\in [n]} p_{x,i} = 1$.

Note that the definition of $\MM_{[1]}(f)$ is well defined for partial functions provided that there exist $x, y \in f^{-1}(0) \times f^{-1}(1)$ such that $x$ and $y$ differ in exactly one bit. This is equivalent to sensitivity, $\s(f)$, being non-zero.  Aaronson et al.~\cite{ABKRT21} proved the following theorem which also hold for these partial functions.

\begin{theorem}(Thm. 28 in \cite{ABKRT21})\label{thm:MMandlambda}
  For any  Boolean function $f$,
  $\lambda(f) = \MM_{[1]}(f)\;.$
\end{theorem}

Here we consider single-bit versions of $\CGpub$ and $\CG$ and show
that they characterize sensitivity and $\lambda^2$, respectively,
up to constant factors.
\begin{definition}[Single-bit private coin certificate game complexity]
\label{def:CG1}
For any (possibly partial) Boolean function $f$ with $\s(f) \neq 0$
\[\CG_{[1]}(f) := 
\min_{p} \max_{\stackrel{x,y\in f^{-1}(0) \times f^{-1}(1)}
  {d(x,y)=1}} 
\frac{1}{\winning(p;x,y)}
= \min_p \max_{x,i: f(x)=1- f(x^{(i)})}
  \frac{1}{p_{x,i}p_{x^{(i)},i}},
  \]
  where $p$ is a collection of nonnegative
variables $\{p_{x,i}\}_{x,i}$ 
that satisfies, for each $x{\in}\zone^n$,
$\sum_{i\in [n]} p_{x,i}=1,$ and 
$\winning(p;x,x^{(i)})$ is the probability
that both players output the unique index $i$
where $x$ and $x^{(i)}$ differ.
(Note that $\winning(p;x,x^{(i)})= p_{x,i}p_{x^{(i)},i}$.)
\end{definition}
Recall that when the players share randomness, a public-coin randomized
strategy is 
a distribution over pairs $(A,B)$ of deterministic strategies. 
We assign a nonnegative variable $p_{A,B}$ to each strategy and
require that they sum to 1.
We say that a pair of strategies $(A,B)$ is correct on
$x,y$ if $A(x)=B(y)=i$ and $x_i\neq y_i$.

\begin{definition}[Single-bit public coin certificate game complexity]
\label{def:CGpub1}
For any (possibly partial) Boolean function $f$ with $\s(f) \neq 0$
\[\CGpub_{[1]}(f) := 
\min_{p} \max_{\stackrel{x,y\in f^{-1}(0) \times f^{-1}(1)}
  {d(x,y)=1}} 
\frac{1}{\winning^{pub}(p;x,y)}
= \min_p \max_{x,i: x\in f^{-1}(0), x^{(i)}\in f^{-1}(1)}
  \frac{1}{\winning^{pub}(p;x,x^{(i)})},
\]
where $p$ is a collection of nonnegative variables $\{p_{A,B}\}_{A,B}$
satisfying $ \sum_{(A,B)}p_{A,B} = 1$
and 
$ \winning^{pub}(p;x,y) = \sum_{(A,B)\text{ correct on } x,y} p_{A,B}$.
\end{definition}

\begin{theorem}\label{thm:cgpub1sens}
  For any (possibly partial) Boolean function $f: \{0,1\}^n \rightarrow \{0,1\}$ with $\s(f) \neq 0$
  $\CGpub_{[1]}(f) = \Theta( \s(f))\;.$
\end{theorem}

\begin{proof} 

  \noindent
  {\bf Upper bound by sensitivity}
  We use the hashing based approach, similarly to the upper bounds
  on $\CGpub$ by $\C$ and $\EC$ (Section \ref{sec:CGpub-C-EC}).

  Let $S$ be a finite set of cardinality $\s(f)$. 
  An element  $z \in S$ is fixed as part of the specification of the
  protocol ($z$ does not depend on the input).

Using shared randomness, the players select a function
  $h: [n] \rightarrow S$ 
  as follows.
    Let $h: [n] \rightarrow S$ be a random
  hash function such that for each $i \in [n]$, $h(i)$ is selected
  independently and uniformly from $S$.

  For $x \in f^{-1}(0)$ 
  let $A_x$ be the set of indices of the sensitive bits of $x$, that is
  $A_x = \{ i \in [n] | f(x) = 1- f(x^{(i)}\}.$
Similarly, for $y \in f^{-1}(1)$ 
let $B_y = \{ i \in [n] | f(y) = 1- f(y^{(i)}\}$.

After selecting $h$ using shared randomness, the players
proceed as follows.
On input $x$, Alice outputs an index $i \in A_x$ such that $h(i) = z$,
and on input $y$, Bob outputs an index $j \in B_y$ such that $h(j) = z$.
If they have several valid choices, or if they have
no valid choices they output arbitrary indices.

Let $i^* \in A_x \cap B_y$, such that $x_{i^*} \neq y_{i^*}$.
Notice that for $x \in f^{-1}(0)$ and $y \in f^{-1}(1)$ such that
$d(x,y)=1$ there is exactly one such index $i^*$.

Next, we estimate what is the probability that both players output $i^*$.
Recall that by the definition of $h$,
the probability that $h(i^*)=z$
is $\frac{1}{|S|} = \frac{1}{\s(f)}$.
Notice that for any $i \in A_x \cup B_y$
the number of elements different from $i$ in $A_x \cup B_y$ is
$\ell = |A_x \cup B_y| -1 \leq 2(|S|-1)$,
since
$\max \{|A_x|, |B_y| \} \leq \s(f) = |S|$.
Thus
for any $z \in S$ and any $i \in A_x \cup B_y$
the probability (over the choice of $h$)
that no element other than $i$
in $A_x \cup B_y$
is mapped to $z$ by $h$ is
$( 1 - \frac{1}{|S|})^{\ell} \geq \frac{1}{e^2}$.

Thus, the players output a correct answer with probability at least
$\frac{1}{e^2} \frac{1}{\s(f)}$.

 {\bf Lower bound by sensitivity}
We will use the
dual formulation of $\CGpub_{[1]}$ obtained  similarly to
Proposition \ref{prop:CGpub-dual}. The only difference is that the distribution
$\mu$ takes nonzero values only on pairs $x, x^{(i)}$ (on pairs with
Hamming distance 1).
Let $x^*$ be an input such that $\s(f;x^*) = \s(f) =:s$, and assume without
loss of generality that $f(x^*)=0$.
Consider the following distribution $\mu$  over input pairs at
Hamming distance 1.
$\mu_{x^*,y} = \frac{1}{s}$ for $y \in f^{-1}(1)$ such that $d(x^*,y) = 1$
and $\mu_{x^*,y}=0$ for every other $y$.
Furthermore, $\mu_{x',y}=0$ for any $y$ and $x' \neq x^*$.
Thus, we only have $s$ input pairs with nonzero measure.

Let $A,B$ be any pair of deterministic strategies for Alice and Bob.
Since $A$ is a deterministic strategy, Alice will output the same
index $i$ for every pair $x^*,y$.
This means that the probability over $\mu$ that the players win
is at most $\frac{1}{s(f;x)} = \frac{1}{s}  = \frac{1}{\s(f)}$ for any pair of deterministic
strategies. 
  \end{proof}

  We define single-bit versions of $\FC$ and $\EC$,
  and show that both are equal to sensitivity.

  \begin{definition}\label{def:singleFC}
  For any (possibly partial) Boolean function $f$ with $\s(f) \neq 0$,
  \begin{itemize}
      \item $\FC_{[1]}(f) = \max_{x\in \zone^n} \, \FC_{[1]}(f,x),$ where
$\FC_{[1]}(f,x) =\min_v \sum_i v_{x,i},$
subject to $ v_{x,i} \geq 1$ for all $i$ such that $f(x)=1- f(x^{(i)})$,
with $v$ a collection of variables $v_{x,i} \geq 0$.
\item   $\EC_{[1]}(f) = \min_{w} \max_{x} \sum_{i\in [n]}  w_{x,i}$,  with $w$ a collection of variables $0\leq w_{x,i}\leq 1$ satisfying
    $w_{x,i}w_{x^{(i)},i} \geq 1$ for all $x,i$ s.t. $f(x) = 1- f(x^{(i)})$.
  \end{itemize}
 
    \end{definition}

  \begin{proposition}
    For any (possibly partial) Boolean function $f: \{0,1\}^n \rightarrow \{0,1\}$ with $\s(f) \neq 0$,
    $\s(f) = \FC_{[1]}(f) = \EC_{[1]}(f)\;.$
  \end{proposition}
  \begin{proof}
    We can think of the values $v_{x,i}$ and $w_{x,i}$
    as weights assigned to the edges of the Boolean hypercube.
    We say that an edge $(x, x^{(i)})$ is sensitive
    (with respect to the function $f$)
    if $f(x) = 1- f(x^{(i)})$.
  First notice, that both definitions require to place weight at least 1
  on each sensitive edge, thus both $\FC_{[1]}(f)$ and  $\EC_{[1]}(f)$
  are at least $s(f)$.
  On the other hand, placing weight 1 on each sensitive edge and weight 0
  on every other edge satisfies the constraints of both definitions,
  thus both $\FC_{[1]}(f)$ and  $\EC_{[1]}(f)$
  are at most $s(f)$.
  \end{proof}

Thus we get the following.

\begin{corollary}
  For any (possibly partial) Boolean function $f: \{0,1\}^n \rightarrow \{0,1\}$ with $\s(f) \neq 0$,
  $\s(f) = \FC_{[1]}(f) = \EC_{[1]}(f) = \Theta(\CGpub_{[1]}(f)) \,.$
\end{corollary}

In case of the single-bit version of private coin certificate game complexity we have: 

\begin{theorem}\label{thm:cg1lambda2}
  For any (possibly partial) Boolean function $f: \{0,1\}^n \rightarrow \{0,1\}$ with $\s(f) \neq 0$, 
  $\CG_{[1]}(f) = \lambda^2 \;.$
\end{theorem}

\begin{proof}
  Comparing the definitions of $\MM_{[1]}$ and $\CG_{[1]}$
  (e.g. the formulation of $MM_{[1]}$ in Equation~\ \eqref{def:MM1two} with
  Definition \ref{def:CG1})
  notice that
  $\sqrt{\CG_{[1]}} = \MM_{[1]}$.
  (One can also restate Definition \ref{def:CG1}
  with weights as in Proposition \ref{prop:sqrt-CG} and
  compare that version with the formulation of $\MM_{[1]}$ in Equation~\
  \eqref{def:MM1one}.)
  The statement then follows from Theorem \ref{thm:MMandlambda}.
\end{proof}


\section{Separations between classical adversary and randomized query complexity for partial functions}

For a partial function~$f$, it is even possible to have an exponential separation between $\R(f)$ and $\CGpub(f)$.

\begin{lemma}
\label{lemma:sep_R_and_CGpub}
There is a partial Boolean function $f\colon\{0,1\}^n\to\{0,1,*\}$ such that
\[
\R(f)\geq \Omega(n)
\qquad\text{but}\qquad
\CGpub(f)\leq O(1).
\]
\end{lemma}
\begin{proof}
Fix any error-correcting code $C\subseteq\{0,1\}^n$ of constant rate and constant relative distance, that is, $|C|\geq 2^{\Omega(n)}$ and for every distinct $x,y\in C$ we have that $x$ and $y$ differ in~$\Omega(n)$ coordinates. (For example, we can use a Justesen code or a random code.) Note that any partial function~$f\colon\{0,1\}^n\to\{0,1,*\}$ with domain $f^{-1}(\{0,1\})=C$ has $\CGpub(f)=O(1)$. Indeed, both players simply output a uniform random coordinate.

Finally, we show that if $f\colon C\to\{0,1\}$ is chosen uniformly at random, then $\R(f)\geq \Omega(n)$ with high probability. Indeed, let $\mathcal{T}$ be a randomized decision tree of depth $d$ computing~$f$. That is, $\mathcal{T}$ is a probability distribution over deterministic depth-$d$ decision trees. By standard randomness sparsification techniques (e.g., Newman's theorem~\cite{Newman91}) we may assume that~$\mathcal{T}$ is a uniform distribution over $n^{O(1)}$ many deterministic trees. Each tree $T\in\supp(\mathcal{T})$ can be encoded as a $O(\binom{n}{d})$-bit string (for each leaf of $T$, encode the root-to-leaf path). Hence~$\mathcal{T}$ can be encoded as a string of length $n^{O(1)}\cdot O(\binom{n}{d})$. However, a random function $f$ needs~$\Omega(|C|)=2^{\Omega(n)}$ bits to describe it, with high probability. It follows that $d\geq\Omega(n)$.
\end{proof}

An example of an explicit function to separate $\R$ and $\CGpub$ is the approximate index function constructed by Ben-David and Blais~\cite{bDB20}. The proof of this exponential separation is given in Appendix~\ref{sec:CGpub-ApIndex}.

We know that $\CGpub$ and $\FC$ cannot be asymptotically different for a total function. Though, there is a partial function, $\GTH$ (defined by Ambainis et al.~\cite{AKPV}, for definition see Appendix~\ref{sec:FunctionExample}), for which $\FC$ is constant~\cite{AKPV} but $\CGpub$ is $\Theta(n)$ (follows from Theorem~\ref{thm:CGns-cAdv} and $\CMM(\GTH) = \Theta(n)$~\cite{AKPV}).

\section{Relations and separations between measures}


Understanding the relationships between the various models of certificate game complexity would help us understand the power of 
shared randomness over private randomness and the power of quantum shared entanglement over shared randomness in the context of 
certificate games.

The first natural separation to consider is the relation between $\CG$ and $\CGpub$.

\begin{corollary}\label{cor:CG_CGns}
For any total Boolean function $f$,
$\CGpub(f) \leq \CG(f)\leq O(\CGpub(f)^{3})$.
\end{corollary}
\begin{proof} 
The first inequality follows from the definitions and the second inequality follows from 
\[\CG \leq O(\EC(f)^2) \leq O(\FC^2(f)\cdot\s(f)) \leq O(\CGns(f)^2\cdot\s(f)) \leq O(\CGns(f)^3),\]
where the first inequality follows from Theorem~\ref{thm:CG-bounds}, the second was proved in \cite{JKKLSSV20} and the last two inequality follows from Theorem~\ref{thm:CGns-cAdv}.
\end{proof} 

Note that the above corollary follows from $\CG\leq O(\CGpub(f)^2 \cdot \s(f))$ (Theorem~\ref{thm:CG-bounds})

\vspace{.5em} 

\Open{Is there a $c< 3$ such that $\CG(f)\leq O(\CGpub(f)^{c})$?}

\vspace{.5em}  

There are total functions~$f$, for which $\CG(f) = \Theta(\CGpub(f)^2)$.  One such example is the $\Tribes$ function. For $\Tribes_{\sqrt{n}, \sqrt{n}} := \OR_{\sqrt{n}}\circ \AND_{\sqrt{n}}$,   we have $\CGpub(\Tribes_{\sqrt{n}, \sqrt{n}})
    = \Theta(\sqrt{n}),$ and $\CG(\Tribes_{\sqrt{n}, \sqrt{n}}) = \Theta(n)$.

\begin{proof} 
Firstly, note that since the functions $\OR$ and $\AND$ has full sensitivity, from Theorem~\ref{thm:CGns-cAdv} we have $\CGpub(\OR_{\sqrt{n}}) = \CGns(\OR_{\sqrt{n}}) = \Theta(\sqrt{n})$. 

Also, the sensitivity of $\Tribes_{\sqrt{n}, \sqrt{n}}$ is $\Theta(\sqrt{n})$
and hence from Theorem~\ref{thm:CGns-cAdv} we have that the $\CGpub$ and $\CGns$ of $\Tribes_{\sqrt{n}, \sqrt{n}}$ is $\Omega(\sqrt{n}$. The upper bound follows Theorem~\ref{thm:CGpub-EC} and the fact that the certificate complexity of $\Tribes_{\sqrt{n}, \sqrt{n}}$ is at most $\sqrt{n}$. But we have also provided a separate proof (Theorem~\ref{thm:tribes}) for the upper bound of the  $\Tribes_{\sqrt{n}, \sqrt{n}}$. Thus we have 
$\CGns(\Tribes_{\sqrt{n}, \sqrt{n}})= \CGpub(\Tribes_{\sqrt{n}, \sqrt{n}})
= \Theta(\sqrt{n}).$ 

Now for the certificate game complexity with shared randomness, from Theorem~\ref{thm:CG-bounds} we know that $\CG$ is bounded below by $\R_0$ 
and we know that $\R_0(\Tribes_{\sqrt{n}, \sqrt{n}}) = \Theta(n)$. On the other, 
Theorem~\ref{thm:CG-bounds}  also helps us to upper bound $\CG$ by $(\EC)^2$, and since $\EC(\Tribes_{\sqrt{n}, \sqrt{n}}) \leq \C(\Tribes_{\sqrt{n}, \sqrt{n}}) \leq \sqrt{n}$, so  we have that $\CG(\Tribes_{\sqrt{n}, \sqrt{n}}) = \Theta(n)$.
\end{proof}

The $\Tribes$ function also demonstrates a quadratic separation between $\CGpub$ and $\R$ while showing that the $\CGpub$ measure does not compose.
Also note that any function with $\lambda(f) = n$, like the parity function, demonstrates a quadratic gap between $\CG$ and $\CGpub$. This is because $\CG(f) = \Omega((\pAdv(f))^2)$, from Theorem~\ref{thm:CG-bounds}, and $\pAdv(f) = \Omega(\lambda(f))$. Thus for any such functions $\CG$ is $\Theta(n^2)$ while $\CGpub$ is $\Theta(n)$.

    One possible attempt to tighten the relation between $\CG$ and $\CGpub$ is to
modify the inequality $\CG = O(\EC^2)$. We observe that the bound $\CG(f) \leq O(\EC(f)^2)$ is indeed tight (Parity function). Though, we could possibly find a better relation between $\EC$ and $\CGpub$. Since $\EC(f) \leq C(f)$, we know that $\EC(F) = O(\CGpub(f)^2)$.

\vspace{.5em} 

\Open{What is the minimum $c$ such that $\EC(f)\leq O(\CGpub(f)^{c})$?}

\vspace{.5em}

Note that if $\CGpub = \Theta(\EC)$ we have $R_0\leq O((\CGpub)^2) = O(\FC^2)$, which is a well-known open problem~\cite{JKKLSSV20}. Theorem~\ref{thm:CG-bounds} shows that $\R_0(f) \leq O(\CG(f))$ (though parity shows that these two measures need not be equal). A quadratic bound on $\CG$ with respect to $\CGpub$ will also settle the well-known open problem mentioned above.

~\\
Another possible direction to tighten the relation between $\CG$ and $\CGpub$ is to improve the inequality $\CG = \Omega(\pAdv^2)$. 
\vspace{.5em} 

\Open{What is the biggest separation between $\CG(f)$ and $\pAdv(f)$?}

\vspace{.5em}

To the best of our knowledge, the best upper bound on $\CG$ for total functions in terms of $\pAdv$ is \[\CG \leq O(\FC^2\s) \leq O(\pAdv^{6}),\] 
where the final inequality follows from the fact that $\FC \leq \pAdv^2$\cite{ABK21} and $\s \leq \lambda^2\leq \pAdv^2$.
The biggest separation between $\CG$ and $\pAdv$ in this direction is cubic: there is a total Boolean function $f$ for which $\CG(f) \leq \Omega(\EC(f)^{3/2})$. In \cite{ABBLSS} they constructed a ``pointer function" $g$, for which $\R_0(g) = \Omega(\Q(g)^3)$.  We observe that, for the pointer function, \[\CG(g) \geq \Omega(\R_0(g)) \geq \Omega(\Q(g)^3) \geq \Omega(\pAdv(g)^3),\]
where the first inequality follows from Theorem~\ref{thm:CG-bounds} and the other inequalities follows from earlier known results. This separation can also be achieved by the cheat sheet version of $k-$Forrelation function that gives a cubic separation between $\Q$ and $\R$ \cite{BS21, ABK16}.

However (from Theorem~\ref{thm:CG-bounds}) for any total Boolean function $f$,  $(\pAdv(f))^2 \leq O(\CG(f))$ and this inequality is in fact tight (for any total function with full spectral sensitivity, such as parity). In fact, the two quantities, $\CG$ and $(\pAdv)^2$, are asymptotically identical for symmetric functions~\cite{MNP21}.

~\\
Another upper bound on $\CG$ that we observe is $\CG \leq \C^0\cdot \C^1$. While for some functions (like the $\Tribes$ function) the two quantities  
$\CG$ and $\C^0\cdot \C^1$ are asymptotically equal we note that there are functions for which $\CG$ is significantly less than $\C^0\cdot \C^1$. 
\begin{corollary}[\cite{JKKLSSV20, GSS16}]
There exists a total function $f:\{0,1\}^N \to \{0,1\}$ for which, 
$\C^0(f) = \Theta(N)$, $\C^1(f) = \Theta(\sqrt{N})$ and $\EC(f) = \Theta(\sqrt{N})$.
Thus $C^0(f)\cdot C^1(f) = \Omega(\CG(f)^{3/2})$.
\end{corollary}
\begin{proof} 
In~\cite[Theorem 11]{JKKLSSV20} they constructed a total function $f:\{0,1\}^N \to \{0,1\}$ such that $\C^0(f) = \Theta(N)$ and $\C^1(f) = \Theta(\sqrt{N})$ and $\EC(f) = \Theta(\sqrt{N})$. Thus, from Theorem~\ref{thm:CG-bounds} we have
$\CG(f) = \Theta(\EC(f))^2 \leq \Theta(N)$. Thus we have the corollary.
\end{proof} 

The separations between single bit version and general version of certificate games is pretty interesting too. One of the enticing open problems in this area of complexity theory is the sensitivity-block sensitivity conjecture. The best gap between $\bs(f)$ and $\s(f)$ is quadratic: that is there exists a function $f$ such that $\bs(f) = \Theta(\s(f)^2)$. The conjecture is that this is indeed tight, that is, for any Boolean function $f$, $\bs(f) = O(\s(f)^2)$. In the seminal work of \cite{Huang} the degree of a Boolean function was bounded by the square of sensitivity, and this is tight for Boolean functions. Since the degree of a Boolean function is quadratically related to the block sensitivity, we have $\bs(f) \leq O(\s(f)^4$. Unfortunately, this approach via degree will not be able to give any tighter bound on block sensitivity in terms of sensitivity. 

Estimating certificate game complexity may be a possible way to prove a tighter 
bound on block sensitivity in terms of sensitivity.  Given the result in Theorem~\ref{thm:cgpub1sens}, designing a strategy for $\CGpub$ using 
$\CGpub_{[1]}$ may help us solve the sensitivity-block sensitivity conjecture. 

\vspace{.5em} 

\Open{What is the smallest $c$ such that, for any Boolean function $f$, $\CGpub(f) = O(\CGpub_{[1]}(f)^c)$?}

\vspace{.5em}

Note that $\CGpub(f) = O(\CGpub_{[1]}(f)^2)$ is equivalent to $\bs(f) \leq O(\s(f)^2)$ (the separation between $\FC$ and $\s$ is same as $\bs$ and $\s$ by~\cite{KT16}). It may seem too much to expect that the single-bit version
of the game can help get upper bounds on the general public coin setting, but thanks to Huang's breakthrough result \cite{Huang}, we already know that 
$\CGpub(f) = O(\CGpub_{[1]}(f)^4)$ for any total Boolean function $f$.


\subsubsection*{Acknowledgements}

We thank Jérémie Roland for helpful discussions, Chandrima Kayal for pointing out a function which separates $\C$ and $\Q$, and  the anonymous referees  for helpful comments.
RM would like to thank IRIF, Paris for hosting him where part of the work was done. AS has received funding from the European Union’s Horizon 2020 research and innovation program under the Marie Sklodowska-Curie grant agreement No 754362. Additional support comes from the French ANR projects ANR-18-CE47-0010 (QUDATA) and ANR-21-CE48-0023 (FLITTLA) and the QOPT project funded by the European Union's Horizon 2020 Research and Innovation Programme under Grant Agreement no. 731473 and 101017733.
MG is supported by the Swiss State Secretariat for Education, Research and Innovation (SERI) under contract number MB22.00026.
 
\addcontentsline{toc}{section}{Bibliography}



\begin{appendix}

\section{Approximate Index: Exponential gap between \texorpdfstring{$\R$}{R} and \texorpdfstring{$\CGpub$}{CGpub} for a \textit{partial Boolean} function} 
\label{sec:CGpub-ApIndex}

We saw that $\CGpub$ of a Boolean function lies between its randomized query complexity and randomized certificate complexity; the same is true for $\noisyR$.

The measure $\noisyR$ was introduced in \cite{bDB20} (please refer to \cite{bDB20} for the formal definition) to study how randomised query complexity $\R$ behaves under composition and it was shown that $\R(f \circ g) = \Omega(\noisyR(f)\R(g))$. As it was also shown that almost all lower bounds (except $\Q$) on $\R$ are also lower bounds on $\noisyR$, it is interesting to see whether $\CGpub$ is also a lower bound on $\noisyR$. 

\vspace{.5em} 

\Open{Is it the case that for all $f$, $\CGpub(f)\leq O(\noisyR(f))$?}

\vspace{.5em} 

Ben-David and Blais~\cite{bDB20} constructed the approximate index function, which is the only function known where $\noisyR$ and $\R$ are different. But the approximate index function that they construct is not a total Boolean function but a partial Boolean function. 

Let $\ApIndex_k$ be the approximate index function where the input has an address part, say~$a$, of $k$ bits and a table with $2^k$ bits. The function is defined on inputs where all positions of the table labelled by strings within $\frac{k}{2} - \sqrt{k \log{k}}$ Hamming distance from $a$ have the same value (either 0 or 1), and all positions that are farther away from $a$ have 2 in them, i.e. 

\begin{definition}
\label{def:apIndex}
$\ApIndex_k :\{ 0,1\}^k \times \{ 0,1,2\}^{2^k} \rightarrow \{ 0,1,\ast \} $ is defined as
 \begin{align*}
 \ApIndex_k(a,x) =\begin{cases}
    x_a  &\text{if } x_b = x_a \in  \{ 0, 1\} \text{ for all } b \text{ that satisfy } |b-a|  \leq  \frac{k}{2} - \sqrt{k \log{k}}\\
   & \text{and }x_b = 2 \text{ for all other $b$,} \\
    \ast &\text{otherwise.}
\end{cases}
 \end{align*}
 \end{definition}
Note that, even though the range of $\ApIndex_k$ (as defined above) is non-Boolean, it can be converted into a Boolean function by encoding the input appropriately. This will only affect the lower/upper bounds by a factor of at most two. 
 
Ben-David and Blais showed that $\noisyR(\ApIndex_k) = O (\log k)$, and $\R(\ApIndex) = \Theta (\sqrt{ k \log k})$.
As an indication that $\CGpub$ could be a lower bound on $\noisyR$, we show the following theorem. 

\begin{theorem} \label{thm:CGpub-ApIndex}
The public coin certificate game complexity of $\ApIndex$ on $n = k + 2^k$ bits is
$\CGpub(\ApIndex_k) = O(\log{k})$.
\end{theorem}

\begin{proof}[Sketch of Proof of Theorem~\ref{thm:CGpub-ApIndex}]

We can use the hashing framework to show an exponential separation between $\R$ and $\CGpub$ for Approximate Index, a partial function.

A central ingredient to the proof of this theorem is the following lemma that captures yet another application of the hashing based framework introduced in Section~\ref{sec:framework}
(we state it in a more general form).

\begin{lemma}\label{lemma:hashing}
  Let $L$ be an integer.
  Assume that for every $x \in f^{-1}(0)$ and $y \in f^{-1}(1)$
  there are sets $A_x$ depending only on $x$, and $B_y$ depending only on $y$, of size $L$, 
    such that 
  any element of $A_x \cap B_y$ is a correct output
  on the input pair $(x,y)$, i.e. for any
  $i \in A_x \cap B_y$, we have $x_i \neq y_i$.
  If for any $x \in f^{-1}(0)$ and $y \in f^{-1}(1)$,
  $L=|A_x|=|B_y| \leq t |A_x \cap B_y| \;,$
  then $\CGpub(f) \leq O(t^2)$.
\end{lemma}

\begin{proof} 
  Let $A_x$ and $B_y$ be sets of size $L$ guaranteed by the statement of
  the lemma.
  We can assume that for $t$ in the statement of the lemma
  $20 \leq t \leq 0.1 L$ holds,
  since $O(L^2)$ is a trivial upper bound on $\CGpub(f)$.
    Let $S$ be a finite set with $|S| = \lfloor \frac{L}{2t} \rfloor > 1$.
    Let $z$ be a fixed element of $S$ (e.g. the first element of $S$)
    given as part of the specification of the protocol.
    (Note that $z$ could also be selected using shared randomness, but
    this is not necessary.)
    
Let $T \subseteq [n]$ be a set of possible outputs that contains
  the sets $A_x$ and $B_y$ for every $x \in f^{-1}(0)$ and $y \in f^{-1}(1)$.
  Let $h: T \rightarrow S$ be a random
  hash function such that for each $i \in T$, $h(i)$ is selected
  independently and uniformly from $S$.
  The players select such $h$ using shared randomness.
  Then, on input $x$, Alice outputs a uniformly random element
  from  $h^{-1}(z) \cap A_x$(if this set is empty, she outputs an arbitrary element).
  On input $y$,  Bob outputs a uniformly random element of
  $h^{-1}(z) \cap B_y$
(if this set is empty, he outputs an arbitrary element).
 
\begin{claim}
  For any $x \in f^{-1}(0)$ and $y \in f^{-1}(1)$, 
  $$Pr[h^{-1}(z) \cap A_x \cap B_y = \emptyset] \leq \frac{1}{e^2}$$
  where the probability is over the choice of the hash function $h$.
\end{claim}

\begin{proof}
  Notice that our setting  implies that for
  any $x \in f^{-1}(0)$ and $y \in f^{-1}(1)$,
  $|A_x \cap B_y| \geq \frac{L}{t} \geq 2|S|$. Thus,  
  $Pr[h^{-1}(z) \cap A_x \cap B_y = \emptyset] =
  (1 - \frac{1}{|S|})^{|A_x \cap B_y|} \leq (1 - \frac{1}{|S|})^{2|S|} \leq 
  \frac{1}{e^2}$.
\end{proof}

\begin{claim}
  For any $x \in f^{-1}(0)$ and $y \in f^{-1}(1)$, 
  $Pr[|h^{-1}(z) \cap A_x| > 3t ] \leq \epsilon$ and
  $Pr[|h^{-1}(z) \cap B_y| > 3t ] \leq \epsilon$,
  where $\epsilon = e^{-0.1 t}$.
\end{claim}
\begin{proof} 
  Notice that the expected size (over the choice of the hash function $h$)
  of the intersection of the pre-image of $z$ with the set $A_x$ is
  $E[|h^{-1}(z) \cap A_x|] = \frac{|A_x|}{|S|} \leq 2.1 t $.
  The claim follows by using the following form of the
  Chernoff bound \cite{MU}:
  $$Pr[X \geq (1 + \delta) \mu] \leq e^{- \frac{\delta^2 \mu}{2 + \delta}}$$
  where X is a sum of independent random variables with values from $\{0,1\}$
  and $\mu = E[X]$.
  The proof with respect to $B_y$ is identical.
\end{proof}

Using the above two claims, we obtain that with probability at least
$1 - e^{-2} - 2e^{-0.1 t} > \frac{1}{2}$ the following conditions hold:

(i) $h^{-1}(z) \cap A_x \cap B_y \neq \emptyset $ and

(ii) $h^{-1}(z)  \cap A_x$ and $h^{-1}(z) \cap B_y$ are both nonempty
and have size at most $3t$.

Let $i^* \in h^{-1}(z) \cap A_x \cap B_y$.
Then $i^*$ is a correct output, and the probability that
both Alice and Bob select $i^*$ as their output is at least
$\frac{1}{9t^2}$.
Thus on any input $x \in f^{-1}(0)$ and $y \in f^{-1}(1)$,
the players output a correct answer with probability at least
$\frac{1}{18 t^2}$.

\end{proof}

Before we see how the hashing lemma helps prove Theorem~\ref{thm:CGpub-ApIndex}, we define the following notation. 
 
The Hamming Sphere of radius $r$ centred at a $k$-bit string $a$, denoted as $\HSph_a(r)$, contains all strings $z \in \{0,1\}^k$ that are at distance exactly $r$ from $a$. Similarly the Hamming Ball of radius $r$ centred at $a$, denoted as $\HBall_a(r)$, contains all strings $z \in \{0,1\}^k$ such that $d(a,z)\leq r$. For the $\ApIndex_k$ function, a valid input has the function value in all positions in the table indexed by strings in $\HBall_a\left(\frac{k}{2} - \sqrt{k \log{k}}\right)$ where $a$ is the address part.

\begin{proof} [Proof of Theorem~\ref{sec:CGpub-ApIndex}]

We consider two different strategies for different kinds of inputs: the first for when the Hamming distance between the address parts $a, b$ of the inputs is large, i.e. $d(a,b) \geq k/\log{k}$ and the second when the distance is smaller. For the first case, Alice and Bob use public randomness to sample an index $i \in [k]$ and this bit differentiates $a$ from $b$ with probability $\geq 1/\log{k}$. In the other case, we first show that $\Omega(1/\sqrt{\log{k}})$ fraction of the Hamming Ball $\HBall_a\left(\frac{k}{2} - \sqrt{k \log{k}}\right)$ around $a$ (or $b$) intersects that around $b$ (or $a$). We then use the hashing lemma (Lemma~\ref{lemma:hashing}) for Alice and Bob to pick an index in the intersection with probability $\Omega(1/\log{k})$.

\paragraph*{Public coin strategy for $\ApIndex$:}
Let us suppose that Alice has an input $(a,x) \in f^{-1}(1)$ and Bob has $(b,y) \in f^{-1}(0)$. We will consider two separate strategies for Alice and Bob to win the public coin Certificate Game with probability $\Omega(\frac{1}{\log{k}})$. They choose to play either strategy with probability $1/2$. 
\begin{itemize}
\item {\bf Strategy 1: }
\begin{tcolorbox}
Alice and Bob sample a random element $z \in [k]$ using public coins and output the element $z$.
\end{tcolorbox}
This strategy  works for inputs for which the Hamming distance between the address parts $a$ and $b$ is large, i.e. $d(a,b) \geq \frac{k}{\log{k}}$. The probability that this strategy succeeds, 
$\Pr[ a_z \ne b_z] \geq \frac{1}{\log{k}}.$
\item {\bf Strategy 2:}
We use the strategy described in Lemma \ref{lemma:hashing} where $A_{(a,x)}$ and $B_{(b,y)}$ are Hamming Balls of radius $\frac{k}{2} - \sqrt{k\log{k}}$ centred at $a$ and $b$ respectively. Let $S$ be a set of size $\left\lfloor \frac{|A_x|}{2\sqrt{\log{k}}} \right\rfloor$
\begin{tcolorbox}
 \begin{itemize}[label = $\bullet$]
 \item Alice and Bob agree on a $z \in S$ in advance.
\item They sample a random hash function $h : \{0,1\}^k \mapsto S$ using public randomness.
\item Alice outputs a uniformly random element from $h^{-1}(z) \cap A_{(a,x)}$ (if this set is empty, she
outputs an arbitrary element). Similarly, Bob outputs a uniformly random element of $h^{-1}(z) \cap B_{(b,y)}$, and if empty, an arbitrary element.
\end{itemize}
\end{tcolorbox}
 The proof that this strategy works for inputs where the Hamming distance between the address parts $a$ and $b$ is small, i.e. $d(a,b) \leq \frac{k}{\log{k}}$ essentially relies on the following lemma. 
\end{itemize}

\begin{lemma}{(Intersection Lemma): } \label{lem:int_surface}
For two $k-$bit strings $a$ and $b$ at Hamming distance $\frac{k}{\log k}$, a Hamming sphere of radius $r$ centred at $a$  has $\frac{c}{\sqrt{\log{k}}}$ fraction of it lying in the Hamming ball of the same radius centred at $b$
\begin{equation*}
\frac{\lvert \HSph_a\left(r\right) \cap \HBall_b\left(r\right)\rvert}{\lvert \HSph_a\left(r\right) \rvert} \geq \frac{c}{\sqrt{\log k}}
\end{equation*} 
where $\frac{k}{2} - 100\sqrt{k\log{k}}\leq r \leq \frac{k}{2} - \sqrt{k\log{k}}$ and $c$ is a constant.
\end{lemma}

The proof of Lemma~\ref{lem:int_surface} is given in Appendix~\ref{app:int_surface}. The basic outline of the proof is as follows: the fraction $\frac{\lvert \HSph_a\left(r\right) \cap \HBall_b\left(r\right)\rvert}{\lvert \HSph_a\left(r\right) \rvert}$ is at least the sum of probabilities from a hypergeometric distribution $P_{k,m,r}$ from $\frac{m}{2}$ to $m$ where $m = \frac{k}{\log{k}}$ is the distance between the Hamming Ball and the Sphere. We show in Lemma~\ref{lem:symmetric} that the hypergeometric distribution $P_{k,m,r}$ is symmetric about $\frac{m}{2}$ for a range up to $200\sqrt{m}$.  The expected value $E$ of $P_{k,m,r}$ for our choice of $m$ and $r$ lies between $\frac{m}{2} - 100\sqrt{m}$ and $\frac{m}{2} - \sqrt{m}$. We have a concentration bound for hypergeometric distribution $P_{k,m,r}$ by Hoeffding \cite{Hoeffding} stated in Lemma~\ref{lem:concentration} that the sum of the probabilities around the expected value of width $\sqrt{r}$ is at least 0.7. Using the property of hypergeometric distributions that it is monotone increasing up to the expected value $E$ and monotone decreasing beyond it shown in Lemma~\ref{lem:monotonicity}, we derive a concentration bound of width $\sqrt{m}$ around $E$ that the probabilities in this range sum to at least $0.7 \times \frac{\sqrt{m}}{\sqrt{r}}$, which for our choice of $m$ and $r$ is at least $\frac{1}{\sqrt{\log{k}}}$. This gives us that the $\frac{\lvert \HSph_a\left(r\right) \cap \HBall_b\left(r\right)\rvert}{\lvert \HSph_a\left(r\right) \rvert} \geq \frac{c'}{\sqrt{\log{k}}}$ for a constant $c'$.

Since we can show most of the weight of the Hamming ball is concentrated on outer layers (proof of which is given in the Appendix \ref{lem:outersurface}) and since the size of the intersection of the Hamming Balls increases as the distance between them decreases, we easily get the following corollary from Lemma~\ref{lem:int_surface}.

\begin{corollary} \label{cor:int_ball}
For two $k-$bit strings $a$ and $b$ at Hamming distance at most $\frac{k}{\log k}$, the ratio of $k-$bit strings in the intersection between the Hamming balls  of radius $\frac{k}{2} - \sqrt{k\log{k}}$ centred at $a$ and $b$ to the total size of each Hamming Ball is at least $\frac{c_1}{\sqrt{\log k}}$.
\begin{equation*}
\frac{\lvert \HBall_a\left(\frac{k}{2} - \sqrt{k\log{k}}\right) \cap \HBall_b\left(\frac{k}{2} - \sqrt{k\log{k}}\right)\rvert}{\lvert \HBall_a\left(\frac{k}{2} - \sqrt{k\log{k}}\right) \rvert} \geq \frac{c_1}{\sqrt{\log k}}
\end{equation*} 
where $c_1$ is a constant.
\end{corollary}
Using the hashing-based framework described in Lemma~\ref{lemma:hashing} with $A_x = \HBall_a\left(\frac{k}{2} - \sqrt{k\log{k}}\right)$ and $B_y = \HBall_b\left(\frac{k}{2} - \sqrt{k\log{k}}\right)$, we get that $\CGpub(\ApIndex) = O(\log{k})$ as $t = \sqrt{\log{k}}/c$ where $c$ is a constant.
\end{proof}

The analysis of the strategy reduces to a very natural question: what is the intersection size of two Hamming balls of radius $\frac{k}{2} - \sqrt{k \log{k}}$ whose centers are at a distance $\frac{k}{\log{k}}$? We are able to show that the intersection is at least an  $\Omega(\frac{1}{\sqrt{log{k}}})$ fraction of the total volume of the Hamming ball. This result and the techniques used could be of independent interest.

To bound the intersection size, we focus on the outermost $\sqrt{k}$ layers of the Hamming ball (since they contain a constant fraction of the total volume), and show that for each such layer the intersection contains an  $\Omega(\frac{1}{\sqrt{log{k}}})$ fraction of the elements in that layer.

For a single layer, the intersection can be expressed as the summation of the latter half of a hypergeometric distribution $P_{k,m,r}$ from $\frac{m}{2}$ to $m$ ($m = \frac{k}{\log{k}}$ is the distance between the Hamming Balls and $r$ is the radius of the layer). By using the ``symmetric'' nature of the hypergeometric distribution around $\frac{m}{2}$ for a sufficient range of values, 
this reduces to showing a concentration result around the expectation with width $\sqrt{m}$ (as the expectation for our choice of parameters is $ \frac{m}{2} - O(\sqrt{m})$). 

We use the standard concentration bound on hypergeometric distribution with width~$\sqrt{r}$ and reduce it to the required width $\sqrt{m}$ by noticing a monotonicity property of the hypergeometric distribution. 
\end{proof}

Although we have proven an upper bound on $\CGpub(\ApIndex)$, a lower bound has not been shown and we leave it as an open problem.

\vspace{.5em} 

\Open{Give a lower bound on $\CGpub(\ApIndex)$.}
\subsection{Proof of the Intersection Lemma~\ref{lem:int_surface}}
\label{app:int_surface}

The Hamming sphere $\HSph_a(r)$ centred at the $k-$bit string $a$ of radius $r$ contains $\binom{k}{r}$ $k-$ bit strings, i.e.
$\lvert \HSph_a(r) \rvert =  \binom{k}{r}$.

Suppose we denote the Hamming distance between $a$ and $b$ as $m$. For our purposes, we choose  $m = \frac{k}{\log k}$. 
A $k-$bit string $z$ at a distance $r$ from $a$ lies in $\HBall_b(r)$ if on the $m$ indices that $a$ differs from $b$, $z$ is closer to $b$ than $a$. The number of $k-$bit strings at a distance $r$ from $a$ that lie in $\HBall_b(r)$,
\begin{equation*}
\lvert \HSph_a(r) \cap \HBall_b(r)\rvert = \left\lvert\left\{ z\in \{0,1\}^k \mid d_H(a,z) = r \wedge d_H(b,z) \leq r  \right\}\right\rvert \geq \sum_{j = m/2}^m\binom{m}{j}\binom{k-m}{r-j}
\end{equation*} 
The hypergeometric distribution on parameters $k, m $ and $r$, for $0 \leq j \leq m$ is given by,
\begin{equation*}
P_{k,m,r}(j) = \frac{\binom{m}{j}\binom{k-m}{r-j}}{\binom{k}{r}}
\end{equation*} 

\begin{proposition}\label{prop:int_size}
The fraction of the size of the intersection to the size of the Hamming Ball can be expressed as a sum of probabilities from a hypergeometric distribution, 
$$\frac{\lvert \HSph_a(r) \cap \HBall_b(r)\rvert}{\lvert \HSph_a(r) \rvert} \geq \sum_{j=m/2}^m P_{k,m,r}(j)$$
\end{proposition}

The proof relies on following three lemmas about hypergeometric distribution.

\begin{lemma}[Concentration Lemma~\cite{Hoeffding}]
\label{lem:concentration}
For a hypergeometric distribution $P$ with parameters $k, m$ and $r$, 
\begin{align*}
\sum_{i = 0}^{E - \sqrt{r}} P_{k,m,r}(i) &\leq e^{-2}\\
\sum_{i = E + \sqrt{r}}^{r} P_{k,m,r}(i) &\leq e^{-2}\\
\end{align*}
where $E = \frac{mr}{k}$ is the expected value of the distribution $P$.
\end{lemma}

\begin{lemma}[Symmetric Property]
\label{lem:symmetric}
For the hypergeometric distribution with parameters $m = \frac{k}{\log k}$ and $k/2 - c\sqrt{k\log{k}} \leq r \leq k/2 - \sqrt{k\log{k}}$ 
\begin{equation*}
 \frac{P_{k,m,r}(m/2 + j)}{P_{k,m,r}(m/2 - j)} \geq {c'}
\end{equation*}
where $0 \leq j \leq 2c\sqrt{m}$ and $c, c'$ are constants. 
\end{lemma}
\begin{proof}
From the definition 
\begin{align*}
 \frac{P_{k,m,r}(m/2 + j)}{P_{k,m,r}(m/2 - j)} &= \frac{\binom{m}{m/2 + j}\binom{k-m}{r - m/2 - j}}{\binom{m}{m/2 - j}\binom{k-m}{r - m/2 +j }} \\
 &= \frac{(r - m/2 - j+1)\cdots(r - m/2 + j) }{(k-m/2 - r -j +1)\cdots(k - m/2 - r + j)}\\
 &\geq \left(\frac{r - m/2 - j}{k - m/2 - r + j}\right)^{2j} = \left(1 - \frac{k - 2r + 2j}{k - m/2 - r + j} \right)^{2j}
\end{align*}
where in the last line we have approximated all the terms in the numerator by a factor smaller than the smallest factor and  in the denominator by the largest factor.
On substituting the values for $m$ and $r$, we have
\begin{align*}
 \frac{P_{k,m,r}(m/2 + j)}{P_{k,m,r}(m/2 - j)} &\geq 
 \left(1 - \frac{1}{2j}\left(\frac{2j \left(2 c\sqrt{k\log{k}} + 2j\right)}{k/2 - \frac{k}{2\log{k}} + \sqrt{k\log{k}} + j}\right)\right)^{2j}\\
 &\approx e^{-\left(\frac{2j \left(2 c\sqrt{k\log{k}} + 2j\right)}{k/2 - \frac{k}{2\log{k}} + \sqrt{k\log{k}} + j}\right)} \geq e^{-16c^2} 
\end{align*}
We get the last inequality after replacing $j$ by the largest possible value that we consider which is $2c\sqrt{m}$ and we get $c' \approx e^{-16c^2} $.
\end{proof}

\begin{lemma}[Monotonicity Property]
\label{lem:monotonicity}
For the hypergeometric distribution where $k$ is large and $m = \frac{k}{\log k}$ and $k/2 - c\sqrt{k\log{k}} \leq r \leq k/2 - \sqrt{k\log{k}}$ , $P_{k,m,r}(j+1) \geq P_{k,m,r}(j)$ for $ j \leq E - 1/2 $ and $P_{k,m,r}(j+1) \leq P_{k,m,r}(j) $ otherwise. Here, $E = \frac{mr}{k}$ is the expected value of the distribution $P$.
\end{lemma}
\begin{proof}
From the definition of hypergeometric distribution, we have
\begin{align*}
\frac{P_{k,m,r}(j+1)}{P_{k,m,r}(j)} &= \frac{\binom{m}{j+1}\binom{k-m}{r-j-1}}{\binom{m}{j}\binom{k-m}{r-j}} = \frac{(m-j)(r-j)}{(j+1)(k-m-r+j+1)}
\end{align*} 
If $P_{k,m,r}(j+1) \geq P_{k,m,r}(j) $, we have $
 \frac{(m-j)(r-j)}{(j+1)(k-m-r+j+1)} \geq 1.
$
On simplifying this expression, we get
$
j \leq \frac{mr+ m -k+r-1}{(k+2) } .
$
Similarly we have $P_{k,m,r}(j+1) \leq P_{k,m,r}(j) $ when $j \geq \frac{mr+m -k+r-1}{(k+2) } $. When $k$ is large, $k+2 \approx k$ and $\frac{mr+m -k+r-1}{(k+2) } \approx E - (1 - \frac{m+r}{k})$. On substituting for $m$ and $r$, we get $\frac{m+r}{k} \approx 1/2 + \epsilon$ where $\epsilon \ll 0$. Thus we can conclude that when $k$ is large enough, $P_{k,m,r}(j+1) \geq P_{k,m,r}(j) $ when $j \leq E - 1/2$ and $P_{k,m,r}(j+1) \leq P_{k,m,r}(j) $ otherwise.
\end{proof}

We can now prove the main result of this section.
\begin{proof}[Proof of Lemma~\ref{lem:int_surface}]
To prove this theorem, from Proposition~\ref{prop:int_size} it is enough to show that 
\begin{equation*}
\sum_{j=m/2}^m P_{k,m,r}(j) \geq \frac{c'}{\sqrt{\log k}}
\end{equation*} 
when $m = \frac{k}{\log k}$ and $k/2 - c\sqrt{k\log{k}} \leq r \leq k/2 - \sqrt{k\log{k}}$.
 From the monotonicity property in Lemma~\ref{lem:monotonicity}, we have that 
\begin{equation*}
\sum_{j=E-\sqrt{m}}^{j=E+\sqrt{m}} P_{k,m,r}(j) \geq \frac{\sqrt{m}}{\sqrt{r}}\ \sum_{j=E-\sqrt{r}}^{j=E+\sqrt{r}} P_{k,m,r}(j) 
 > \sqrt{\frac{2}{\log{k}}}\ \sum_{j=E-\sqrt{r}}^{j=E+\sqrt{r}} P_{k,m,r}(j) 
\end{equation*}
From Lemma~\ref{lem:concentration}, we have 
\begin{equation*}
\sum_{j=E-\sqrt{r}}^{j=E+\sqrt{r}} P_{k,m,r}(j) \geq 0.72
\end{equation*}
This gives,
\begin{equation*}
\sum_{j=E-\sqrt{m}}^{j=E+\sqrt{m}} P_{k,m,r}(j) > \sqrt{\frac{2}{\log{k}}} \times 0.72 > \frac{1}{\sqrt{\log{k}}}
\end{equation*}
 For our choice of $m$ and $r$, we have the expected value $m/2- c\sqrt{m} \leq E \leq m/2- \sqrt{m}$. Using Lemma~\ref{lem:symmetric}, by the symmetric property of the hypergeometric distribution for our choice of $m$ and $r$, on reflecting about $m/2$ we have
\begin{equation*}
\sum_{j=m/2}^m P_{k,m,r}(j)  \geq {c'} \sum_{j=E-\sqrt{m}}^{j=E+\sqrt{m}} P_{k,m,r}(j) \geq \frac{c'}{\sqrt{\log k}}.
\end{equation*} 
where $c' \approx e^{-16c^2}$.
\end{proof}

\subsection{Most of the weight is concentrated on outer surfaces of the Hamming ball}

\begin{lemma}\label{lem:outersurface}
 For a Hamming Ball of radius $r = k/2 - \sqrt{k\log{k}}$, the weight contributed by Hamming Spheres of radius $\leq k/2 - 100\sqrt{k\log{k}}$ is small.
 \begin{equation*}
  \frac{\sum_{i =0}^{\frac{k}{2} - 100\sqrt{k\log{k}}}\ \lvert \HSph_a\left(i\right)\rvert}{\lvert \HBall_a\left(\frac{k}{2} - \sqrt{k\log{k}}\right)\rvert} \leq c_1
 \end{equation*}
where $c_1$ is a constant.
\end{lemma}

\begin{proof}
We would like to show
   \begin{equation*}
  \frac{\sum_{j = 0}^{\frac{k}{2} - 100\sqrt{k\log{k}}}\ \binom{k}{j}}{\sum_{j = 0}^{\frac{k}{2} - \sqrt{k\log{k}}}\ \binom{k}{j}} \leq c_1
 \end{equation*}
We use the following form of Chernoff Bound \cite{MU},
\begin{equation*}
\Pr\left[X \leq (1-\delta)\mu\right] \leq e^{\frac{\delta^2\mu}{2}}
\end{equation*}
for $0 \leq \delta \leq 1$ and apply it to the binomial distribution with $p = 1/2$ to get $\sum_{j = 0}^{\frac{k}{2} - 100\sqrt{k\log{k}}}\ \binom{k}{j} \leq 2^k k^{-10^4}$. We now use the following lower bound for the tail of the binomial distribution when $p=1/2$ (which is restated slightly from its original form in \cite{MatousekVondrak}).
\begin{equation*}
\Pr\left[X \leq k/2 - \delta \right] \geq \frac{1}{15} e^{-16\delta^2/k}
\end{equation*}
for $\delta\geq 3k/8$. This gives $\sum_{j = 0}^{\frac{k}{2} - \sqrt{k\log{k}}}\ \binom{k}{j} \geq 2^k \frac{1}{15} k^{-16}$. Thus we have
  \begin{equation*}
  \frac{\sum_{j = 0}^{\frac{k}{2} - 100\sqrt{k\log{k}}}\ \binom{k}{j}}{\sum_{j = 0}^{\frac{k}{2} - \sqrt{k\log{k}}}\ \binom{k}{j}} \leq \frac{15 k^{-10^4}}{ k^{-16}} \ll c_1
 \end{equation*}
\end{proof}


\section{Examples of functions}
\label{sec:FunctionExample}

Interesting examples of total and partial Boolean functions are very important to understand the relations between various complexity measures. In fact constructing interesting functions is one of the commonly used techniques to prove separation between pairs of measures. A number of interesting functions has been constructed for this purpose (for example \cite{GSS16,ABK16,BS21,Cha05,Rub95}).  In this paper we use some of them to understand the relation between the certificate games measures and others. The various complexity measures for the functions we consider is compiled in Table~\ref{Table:total}.

$\OR$ and Parity ($\Parity$) are one of the simplest functions, probably the first ones to be studied for any complexity measures. The bounds on $\Parity_n$ follow from the observation that $\lambda(\Parity_n) = \Theta(n)$ ($\CG(\Parity_n) = \Theta(n^2)$ follows from Theorem~\ref{thm:CG-bounds}); the bounds on $\OR_n$ follow from $\lambda(f) = \Theta(\sqrt{n})$~\cite{ABKRT21}, $\Q(\OR_n) = O(\sqrt{n})$~\cite{Grover96} and the observation that $\s(\OR_n) = \Theta(n)$ (please refer to Figure~\ref{fig:measures}).

$\Tribes_{m,n} = \OR_m \circ \AND_n$ is a non-symmetric function, made by composing $\OR$ and $\AND$. We use it as an example of a total function where $\R$ and $\CGpub$ are asymptotically different. It can be verified that $\C(\Tribes_{\sqrt{n},\sqrt{n}}) = \Theta (\sqrt{n})$, and $\lambda(\Tribes_{\sqrt{n},\sqrt{n}}) = \Q(\Tribes_{\sqrt{n},\sqrt{n}}) =\Theta (\sqrt{n})$ follows from composition~\cite{ABKRT21, LMRSS11}. $\R(\Tribes_{\sqrt{n},\sqrt{n}}) = \Theta(n)$ is from Jain and Klauck~\cite{JK09}, other measures follow from these observations.

The function $\GSS_1$ is a function defined in \cite{GSS16}. It is defined on $\{0,1\}^{n^2}$. The complexity measures of $\GSS_1$ was computed in \cite{GSS16} and \cite{JKKLSSV20}. The blank spaces indicates that the tight bounds are not known.

\begin{center}
\begin{table}[htbp]
\scalebox{.85}{
\begin{tabular}{|c||c|c|c|c|c|c|c|c|c|c|c|}
\hline
  Function     & $\lambda$         & $\s$          & $\bs$             & $\FC$         &$\MM$              & $\Q$                      & $\CGpub$         & $\R$          & $\EC$         & $\C$          & $\CG$ \\
  \hline 
  \hline
  $\OR_n$  &$\Theta(\sqrt{n})$ & $\Theta(n)$   & $\Theta(n)$       & $\Theta(n)$ & $\Theta(\sqrt{n})$  & $\Theta(\sqrt{n})$   & $\Theta(n)$       & $\Theta(n)$   & $\Theta(n)$   & $\Theta(n)$   & $\Theta(n)$ \\
  \hline
  $\Parity_n$  & $\Theta(n)$       & $\Theta(n)$   & $\Theta(n)$       & $\Theta(n)$ & $\Theta(n)$         & $\Theta(n)$          & $\Theta(n)$       & $\Theta(n)$   & $\Theta(n)$   & $\Theta(n)$   & $\Theta(n^2)$\\
  \hline
  $\Tribes_{\sqrt{n},\sqrt{n}}$ & $\Theta(\sqrt{n})$ & $\Theta(\sqrt{n})$ & $\Theta(\sqrt{n})$ & $\Theta(\sqrt{n})$ & $\Theta(\sqrt{n})$ & $\Theta(\sqrt{n})$ & $\Theta(\sqrt{n})$ & $\Theta(n)$ & $\Theta(\sqrt{n})$ & $\Theta(\sqrt{n})$ & $\Theta(n)$\\
  \hline
  $\GSS_1$ & & $\Theta(n)$ & $\Theta(n)$ & $\Theta(n)$ & & & $\Theta(n)$ & & $\Theta(n)$ & $\Theta(n^2)$ & $O(n^2)$\\
  \hline
  \hline
\end{tabular}}
\caption{Some of the commonly referred total functions and their complexity measures}\label{Table:total}
\end{table}
\end{center}

Regarding partial functions we would like to discuss a couple of them that are used in multiple places in the paper to show separations between measures for partial functions - namely the ``approximate indexing" function and the ``greater than half" function. The know measures for these functions are compiled in the Table~\ref{Table:partial}.

$\ApIndex$ is the approximate indexing function defined by Ben-David and Blais~\cite{bDB20}, we show that $\R$ and $\CGpub$ are exponentially separated for this partial function (we know $\R(\ApIndex) = \Theta(\sqrt{k \,\log k})$~\cite{bDB20} and $\CGpub(\ApIndex) = O(\log k)$, from Section~\ref{sec:CGpub-ApIndex}). Rest of the measures mentioned in the table can be observed easily.

There is a partial function, $\GTH$ (defined by Ambainis et al.~\cite{AKPV}, see Definition~\ref{def:GTH}), for which $\FC$ is constant~\cite{AKPV} but $\CGpub$ is $\Theta(n)$ (follows from Theorem~\ref{thm:CGns-cAdv} and $\CMM(\GTH) = \Theta(n)$~\cite{AKPV}).

\begin{definition}[$\GTH$ \cite{AKPV}]
 \label{def:GTH}
 The ``greater than half" function is a partial function defined only on $n$ bit strings that have Hamming weight 1. The function evaluates to 1 on an input $x$ if the position $i$ where the input bit is 1 is in the second half of the string, i.e. $\GTH: \zone^{n} \rightarrow \zone$ is defined as $\GTH(x) = 1$ if $x_i = 1$ where $i > n/2$.
 \end{definition}

To show that $\CG(\GTH) = \Theta(n)$, we use the version in Proposition~\ref{prop:sqrt-CG}. For a $1$-input $y$, we only put a non-zero weight of $\sqrt{n}$ on index $i$ where $y_i = 1$. For a $0$-input, we put  a non-zero weight of $\frac{1}{\sqrt{n}}$ only on indices $i$ such that $i > n/2$. It can be verified that this is a feasible solution of the equivalent formulation of $\CG$ (from Proposition~\ref{prop:sqrt-CG}) with objective value $n$.

\begin{table}[h!]
\begin{center}
\scalebox{.85}{
\begin{tabular}{|c||c|c|c|c|c|c|c|c|c|c|c|c|}
\hline
   Function     & $\lambda$         & $\s$          & $\bs$             & $\FC$         &$\MM$              & $\Q$                      & $\CMM$ & $\CGpub$         & $\R$          & $\EC$         & $\C$          & $\CG$ \\
   \hline 
   \hline
   $\ApIndex$  & & 0   & $O(1)$ & $O(1)$ &   &    & & $O(\log k)$       & $\Theta(\sqrt{k\,\log k})$   &    & $O(1)$   &  \\
   \hline
   $\GTH_n$  & & 0  & $O(1)$  & $O(1)$ &          &  & $\Theta(n)$ & $\Theta(n)$ &   & $\Theta(n)$   & $O(1)$   & $\Theta(n)$\\
   \hline
\end{tabular}}
\caption{The known complexity measures for $\ApIndex$ and $\GTH_n$ }\label{Table:partial}
\end{center}
\end{table}

\section{$\FC$ as a local version of $\CGpub$}

In this section we will show that $\FC(x)$ can be viewed as a certificate game where Alice's input is fixed. We start with the dual of the $\CGpub$ optimization problem.

For a two-player certificate game $G_f$ corresponding to a (possibly partial) Boolean function $f$,  $ \CGpub(f) = 1/\winning^{pub}(G_f)$ (Proposition~\ref{prop:CGpub-dual}), 
where the winning probability $\winning^{pub}(G_f)$ is given by the following linear program. 
 \begin{align*}  
 \winning^{pub}(G_f)  = &  \quad \min_{\delta, \mu}
\quad  \delta
  \\
  \text{such that }    &
  \sum_{x,y:~A,B \text{ correct on } x,y} 
  \hspace{-2em}
  \mu_{x,y} \leq \delta \quad \text{ for every deterministic strategy } A,B \\
   & \sum_{x,y} \mu_{xy} = 1, \quad  \mu_{x,y} \geq 0,
 \end{align*}
 where $ \mu =\{\mu_{x,y}\}_{x\in f^{-1}(0), \, y \in f^{-1}(1)}$. $A,B$ correct on $x,y$ implies $A(x)=B(x)=i$ and $x_i\neq y_i$. 

Re-normalizing, we get the linear program for $\CGpub(f)$,
 \begin{align*}  
 \CGpub(f)  = & \quad \sum_{x,y} \mu_{x,y}
  \\
  \text{such that }    &
  \sum_{x,y:~A,B \text{ correct on } x,y} 
  \hspace{-2em}
  \mu_{x,y} \leq 1 \quad \text{ for every deterministic strategy } A,B \\
   & \quad  \mu_{x,y} \geq 0,
 \end{align*}

Let us define the local version of $\CGpub$ in two stages: Let $\CGpub_L(f,x)$ be the value of the $\CGpub$ game when one of the party's input is fixed to $x$ (say Alice), then $\CGpub_L(f) = \max_x \CGpub_L(x)$. We will show that $\CGpub_L$ is same as $\fbs$; given that $\CGpub(f) = \Theta(\fbs(f))$, we see that local and global version of $\CGpub$ are same.

The linear program for the local version can be written as:
\begin{align*}  
 \CGpub_L(f,x)  = &  \quad  \sum_{y} \mu_{y}
  \\
  \text{such that }    &
  \sum_{y:B \text{ outputs } i \text { on } y \text{ and } x_i \neq y_i} 
  \hspace{-2em}
  \mu_{y} \leq 1 \quad \forall \text{ deterministic strategies } B, \text{ index } i \\
   & \quad  \mu_{y} \geq 0,
 \end{align*}

Where it is understood that the deterministic strategy for Alice, $A$, answers $i$. Notice that fixing an $i$, the \emph{strictest} constraint is obtained by $B$ which answers $i$ whenever $y_i \neq x_i$. This means we can keep just one constraint for every $i$.
\begin{align*}  
 \CGpub_L(f,x)  = &  \quad  \sum_{y} \mu_{y}
  \\
  \text{such that }    &
  \sum_{y: x_i \neq y_i} 
  \mu_{y} \leq 1 \quad \text{ for all } i \\
   & \quad  \mu_{y} \geq 0,
 \end{align*}

Every $y$ (such that $f(x) \neq f(y)$) has a one to one correspondence with a sensitive block $B$ such that $y = x^{\oplus B}$. The linear program can be simplified to the linear program for $\fbs(f,x)$.
\begin{align*}  
 \CGpub_L(f,x)  = &  \sum_{B} \mu_{B}
  \\
  \text{such that }    &
  \sum_{B: i \in B} 
  \mu_{B} \leq 1 \quad \text{ for all } i \\
   & \quad  \mu_{B} \geq 0,
 \end{align*}

\end{appendix}


\end{document}